\documentclass[aps, prl, twocolumn, superscriptaddress, letterpaper, citeautoscript]{revtex4-2}

\usepackage{bm}
\usepackage{graphicx}
\usepackage{xcolor} % covers color
\usepackage{braket}
\usepackage{amsmath,amssymb,amsfonts,amsthm,mathtools}
\usepackage{enumerate}
\usepackage{enumitem}
\usepackage[colorlinks=true, linkcolor=blue, anchorcolor=black, citecolor=blue,
            filecolor=black, menucolor=black, runcolor=black, urlcolor=blue]{hyperref}
\usepackage[section]{placeins}
\usepackage{diagbox}
\usepackage{titlesec}
\usepackage{tikz-cd}
\usepackage{multirow}
\usepackage[title]{appendix}
\usepackage{pdfpages}
\usepackage{changes}
\usepackage{stackengine}
\usepackage{makecell}
\usepackage{bbm}
\usepackage{algorithm}
\usepackage{algpseudocode}
\usepackage{cleveref}

\usepackage{dcolumn}
\usepackage[normalem]{ulem}
\usepackage[version=4]{mhchem}
\usepackage{siunitx}
\usepackage{float}
\usepackage{subfigure} % SI used this; keep for compatibility

\newtheorem{theorem}{Theorem}

\newtheorem{sectiontheorem}{Theorem}[section]
\renewcommand{\thesectiontheorem}{M\arabic{sectiontheorem}}
\newtheorem{sectionlemma}{Lemma}[section]
\renewcommand{\thesectionlemma}{M\arabic{sectionlemma}}
\newtheorem{sectiondefinition}{Definition}[section]
\renewcommand{\thesectiondefinition}{M\arabic{sectiondefinition}}
\newtheorem{sectionremark}{Remark}[section]
\renewcommand{\thesectionremark}{M\arabic{sectionremark}}
\newtheorem{sectioncorollary}{Corollary}[section]
\renewcommand{\thesectioncorollary}{M\arabic{sectioncorollary}}

\renewcommand{\thesectionexample}{M\arabic{sectionexample}}

\renewcommand{\v}[1]{\boldsymbol{#1}}

\newcommand{\redsection}[1]{%
  \section{\color{red}#1}%
  \begingroup\color{red}%
  \everypar{\global\color{red}\aftergroup\resetcolor}%
}
\def\resetcolor{\color{black}}

\makeatletter
\AtBeginDocument{\let\LS@rot\@undefined}
\makeatother

\makeatletter
\renewcommand\NAT@biblabelnum[1]{#1.}
\makeatother

\AtBeginDocument{\setcounter{tocdepth}{1}}

\newcommand{\beginsupplement}{%
  \setcounter{section}{0}%
  \renewcommand{\thesection}{S\arabic{section}}%

  \setcounter{equation}{0}%
  \renewcommand{\theequation}{S\arabic{equation}}%
  \setcounter{figure}{0}%
  \renewcommand{\thefigure}{S\arabic{figure}}%
  \setcounter{table}{0}%
  \renewcommand{\thetable}{S\arabic{table}}%

  \setcounter{theorem}{0}%
  \renewcommand{\thetheorem}{S\arabic{theorem}}%
  \setcounter{definition}{0}%
  \renewcommand{\thedefinition}{S\arabic{definition}}%
  \setcounter{remark}{0}%
  \renewcommand{\theremark}{S\arabic{remark}}%
  \setcounter{fact}{0}%
  \renewcommand{\thefact}{S\arabic{fact}}%

  \setcounter{sectiontheorem}{0}%
  \renewcommand{\thesectiontheorem}{S\arabic{sectiontheorem}}%
  \setcounter{sectionlemma}{0}%
  \renewcommand{\thesectionlemma}{S\arabic{sectionlemma}}%
  \setcounter{sectiondefinition}{0}%
  \renewcommand{\thesectiondefinition}{S\arabic{sectiondefinition}}%
  \setcounter{sectionremark}{0}%
  \renewcommand{\thesectionremark}{S\arabic{sectionremark}}%
  \setcounter{sectioncorollary}{0}%
  \renewcommand{\thesectioncorollary}{S\arabic{sectioncorollary}}%
  \setcounter{sectionexample}{0}%
  \renewcommand{\thesectionexample}{S\arabic{sectionexample}}%
}

\begin{document}

\title{Exact Universal Characterization of Chiral-Symmetric Higher-Order Topological Phases}

\author{Jia-Zheng Li}
\thanks{These authors contributed equally.}
\affiliation{School of Physics and Technology, Wuhan University, Wuhan 430072, China}
\affiliation{Department of Physics and Center for Theory of Quantum Matter, University of Colorado, Boulder CO 80309, USA}
\author{Xun-Jiang Luo}
\thanks{These authors contributed equally.}
\affiliation{School of Physics and Technology, Wuhan University, Wuhan 430072, China}
\affiliation{Department of Physics, Hong Kong University of Science and Technology, Clear Water Bay, Hong Kong, China}
\author{Fengcheng Wu}
\email{wufcheng@whu.edu.cn}
\affiliation{School of Physics and Technology, Wuhan University, Wuhan 430072, China}
\affiliation{Wuhan Institute of Quantum Technology, Wuhan 430206, China}
\author{Meng Xiao}
\email{phmxiao@whu.edu.cn}
\affiliation{School of Physics and Technology, Wuhan University, Wuhan 430072, China}
\affiliation{Wuhan Institute of Quantum Technology, Wuhan 430206, China}

\begin{abstract}
Utilizing Bott index vectors formulated through a series of polynomials of position operators under open boundary conditions, we establish a universal, rigorous, and complete correspondence between the Bott index vector and topological zero-energy corner states in systems with chiral symmetry. Our framework covers systems of arbitrary shapes, including topological phases that are beyond the characterization by previously proposed invariants such as multipole moments or multipole chiral numbers. A key feature of our approach is its ability to capture the real-space patterns of zero-energy corner states, providing a deeper understanding of higher-order topological phases. We provide a rigorous analytical proof of its higher-order correspondence and sum rules for Bott index vectors under different boundary conditions. To demonstrate the effectiveness of our theory, we examine several model systems with representative patterns of zero-energy corner states that lie outside the scope of previous classification frameworks.

\end{abstract}

\maketitle

\noindent{\large{\bf{Introduction}}}\\
The exploration of higher-order topological phases (HOTPs), encompassing insulators~\cite{QuantizedElectricMultipole2017benalcazar,EnsuremathDimensionalEdge2017song,ReflectionSymmetricSecondOrderTopological2017langbehn,ElectricMultipoleMoments2017benalcazar,HigherorderTopologicalInsulators2018vanmiert,HigherorderTopologicalInsulators2018schindler,HigherorderTopologicalInsulators2018khalaf,SecondorderTopologicalInsulators2018geier,QuantizationFractionalCorner2019benalcazar,HigherOrderTopologyAxion2019xu,SymmetryenforcedChiralHinge2019yue,BoundaryobstructedTopologicalPhases2021khalaf}, semimetals~\cite{TopologicalQuadrupolarSemimetals2018lina,HigherOrderWeylSemimetals2020ghorashi,StrongFragileTopological2020wieder,UniversalHigherorderBulkboundary2022lenggenhager,HigherOrderTopologyMonopole2019wanga}, and superconductors~\cite{HigherorderTopologicalInsulators2018khalaf,SecondorderTopologicalInsulators2018geier,HigherOrderTopologicalOddParity2019yan,WeakpairingHigherOrder2018wang,VortexSurfacePhase2020ghorashi,PhasetunableSecondorderTopological2019franca,MajoranaCornerModes2018dwivedi,MajoranaCornerStates2018liu,PhysRevB.104.104510}, extends beyond the realm of first-order topological phases~\cite{ClassificationTopologicalQuantum2016chiu}, suggesting a generalized bulk-boundary correspondence~\cite{HigherOrderBulkBoundaryCorrespondence2019trifunovic,HigherOrderTopologicalBand2021trifunovic}. According to this correspondence, HOTPs can be classified into intrinsic phases protected by crystalline symmetry and extrinsic phases that depend on boundary properties~\cite{SecondorderTopologicalInsulators2018geier,HigherOrderBulkBoundaryCorrespondence2019trifunovic}. Higher-order multipoles~\cite{QuantizedElectricMultipole2017benalcazar,ElectricMultipoleMoments2017benalcazar,HigherorderTopologicalInsulators2018schindler,PhysRevB.108.235150}, such as the quadrupole $q_{xy}$, have been introduced to characterize the corner states in higher-order topological insulators (HOTIs). A significant advancement is the introduction of the multipole chiral number~\cite{ChiralSymmetricHigherOrderTopological2022benalcazar}, aimed at capturing the higher-order topology of chiral-symmetric systems. Despite these pioneering developments, the theoretical understanding of HOTPs, particularly in chiral-symmetric systems, confronts several persistent challenges. Firstly, establishing a \textbf{rigorous correspondence} between proposed topological invariants (e.g., the quadrupole moment or multipole chiral number~\footnote{Please see Supplemental Information, Note 5.}) and the actual presence of HOTPs remains an open problem, with recent studies highlighting inconsistencies and limitations~\cite{QuadrupoleInsulatorCorner2023tao,PhysRevLett.132.213801,PhysRevB.107.045118,TypeIIQuadrupoleTopological2020yang,10.21468/SciPostPhys.17.2.060,PhysRevResearch.2.043012}. Secondly, while HOTPs are often distinguished as intrinsic (crystalline symmetry protected) or extrinsic (boundary dependent)~\cite{SecondorderTopologicalInsulators2018geier,HigherOrderBulkBoundaryCorrespondence2019trifunovic}, a \textbf{unified characterization} that encompasses both types for systems of arbitrary geometry is yet to be developed. Thirdly, complex systems can exhibit diverse `patterns of corner states'. For instance, consider a scenario where a topological Benalcazar-Bernevig-Hughes model~\cite{QuantizedElectricMultipole2017benalcazar,QuantizedElectricMultipole2017benalcazar} is coupled with a chiral-symmetric system possessing zero-energy corner states (ZECSs) located diagonally in a square system~\cite{SecondorderTopologicalInsulators2018geier,PhysRevLett.124.166804} while preserving chiral symmetry. The combined system exhibits varying numbers of ZECSs localized at different corners. We refer to this variability in the distribution of corner states as the ``pattern of corner states", as illustrated in Fig.~\ref{fig1_configuration_sum}(a). Unlike the well-established classification for first-order topological states~\cite{ClassificationTopologicalQuantum2016chiu}, a comprehensive theoretical framework for \textbf{capturing these intricate patterns} of ZECSs in HOTPs is still lacking.

In this work, we directly address these three longstanding challenges by establishing a universal, rigorous, and complete correspondence that elucidates the spatial topological structure of HOTPs in systems with chiral symmetry.
Our approach is founded upon an analytically proven correspondence, detailed in Eq.~\eqref{eqs:correspondence}, which utilizes a Bott index vector defined through polynomials of position operators under open boundary conditions (OBC).
This framework provides a universal method for characterizing chiral-symmetry-protected HOTPs including its spatial topological information across arbitrary dimensions.
Crucially, its independence from crystalline symmetry allows for the characterization of both intrinsic and extrinsic HOTPs in systems of any shape featuring $m$ corners (where $m\in\mathbb{N}^{\ast}$). 
This universal, rigorous, and complete correspondence between Bott indices and HOTPs in non-interacting chiral-symmetric systems significantly advances their theoretical understanding by resolving fundamental issues highlighted in prior research~\cite{ExactHigherorderBulkboundary2021jung,QuadrupoleInsulatorCorner2023tao,PhysRevLett.132.213801,DifficultiesOperatorbasedFormulation2019ono,TypeIIQuadrupoleTopological2020yang}. Moreover, our framework introduces a sum rule that enables the precise differentiation of topological information across various boundary states. Beyond these theoretical contributions, our work paves the way for practical applications in the design and prediction of HOTPs, holding potential implications for condensed matter physics~\cite{shumiya2022evidence,aggarwal2021evidence}, photonics~\cite{wu2023higher,ota2019photonic}, and acoustics~\cite{PhysRevLett.122.244301,PhysRevLett.129.125502}.

\medskip
\noindent{\bf{Bott index and polynomials of position operators}}\\
We consider a Hamiltonian $H$ on an $n$-dimensional lattice with chiral symmetry, $\Pi H \Pi^{-1}=-H$, where $\Pi$ is a traceless unitary matrix representing the chiral operator (the case of chiral symmetry with non-zero trace is discussed in the Supplemental Information, Note 1~\cite{dai2024topologicalclassificationchiralsymmetry}). We take the eigenbasis of $\Pi$ and rewrite the Hamiltonian,
\begin{equation}
\label{eqs:chiralhami}
    H=\left( \begin{matrix}
	0&		h\\
	h^{\dagger}&		0\\
\end{matrix} \right).
\end{equation}
We denote $q=U_AU_B^{\dagger}$ based on the singular value decomposition of $h$ with size $N$: $h=U_A\Sigma U_B^{\dagger}$, where $U_{\alpha}=\left[\Psi_{\alpha,1},\dots,\Psi_{\alpha,N}\right],\text{ } \alpha\in\{\text{A,B}\}$, and A (B) labels the eigen-spaces of chiral operator $\Pi$ with eigenvalue $+1$ ($-1$). $\Sigma$ is a diagonal matrix with nonnegative elements. The eigenstates of $H$ can be written as $|\Psi_{i}^{\pm}\rangle=(1/\sqrt{2})[\Psi_{A,i},\pm\Psi_{B,i}]^{T}$. Our focus is on insulators, meaning that energy gaps energy gaps are required for all states except those that are localized at corners.

For $n=1$, the problem returns to characterizing the first-order topological phase in the AIII class, which has been extensively studied utilizing the winding number~\cite{ClassificationTopologicalQuantum2016chiu} and the Bott index~\cite{RealspaceRepresentationWinding2021lin}. We focus on the Bott index~\cite{DisorderedTopologicalInsulators2011loring,GuideBottIndex2019loring}.
\begin{equation}
\label{eqs:clasc_bott}
\begin{aligned}
    \nu_{(\text{1D})}&=\operatorname{Bott}\left(e^{2\pi i \frac{X}{L}},q\right),\\
    \operatorname{Bott}\left(e^{2\pi i \frac{X}{L}},q\right)&:=\frac{1}{2\pi i}\operatorname{Tr}\operatorname{log}\left(e^{2\pi i \frac{X}{L}}qe^{-2\pi i \frac{X}{L}}q^{\dagger}\right),
\end{aligned}
\end{equation}
where $X$ is the $x$-direction position operator, $L$ is the system length, $\operatorname{log}$ is the principal logarithm. Considering $q$ as a representation of the states and ${X}/{L}$ as a polynomial that functions as a topology-measuring ruler, a critical question arises: Is it possible to formulate measurement polynomials to universally and exactly capture the $n$-th order topological phases for an arbitrary $n$-dimensional Hamiltonian with chiral symmetry?

\begin{figure}[t]
	\includegraphics[width=1.0\columnwidth]{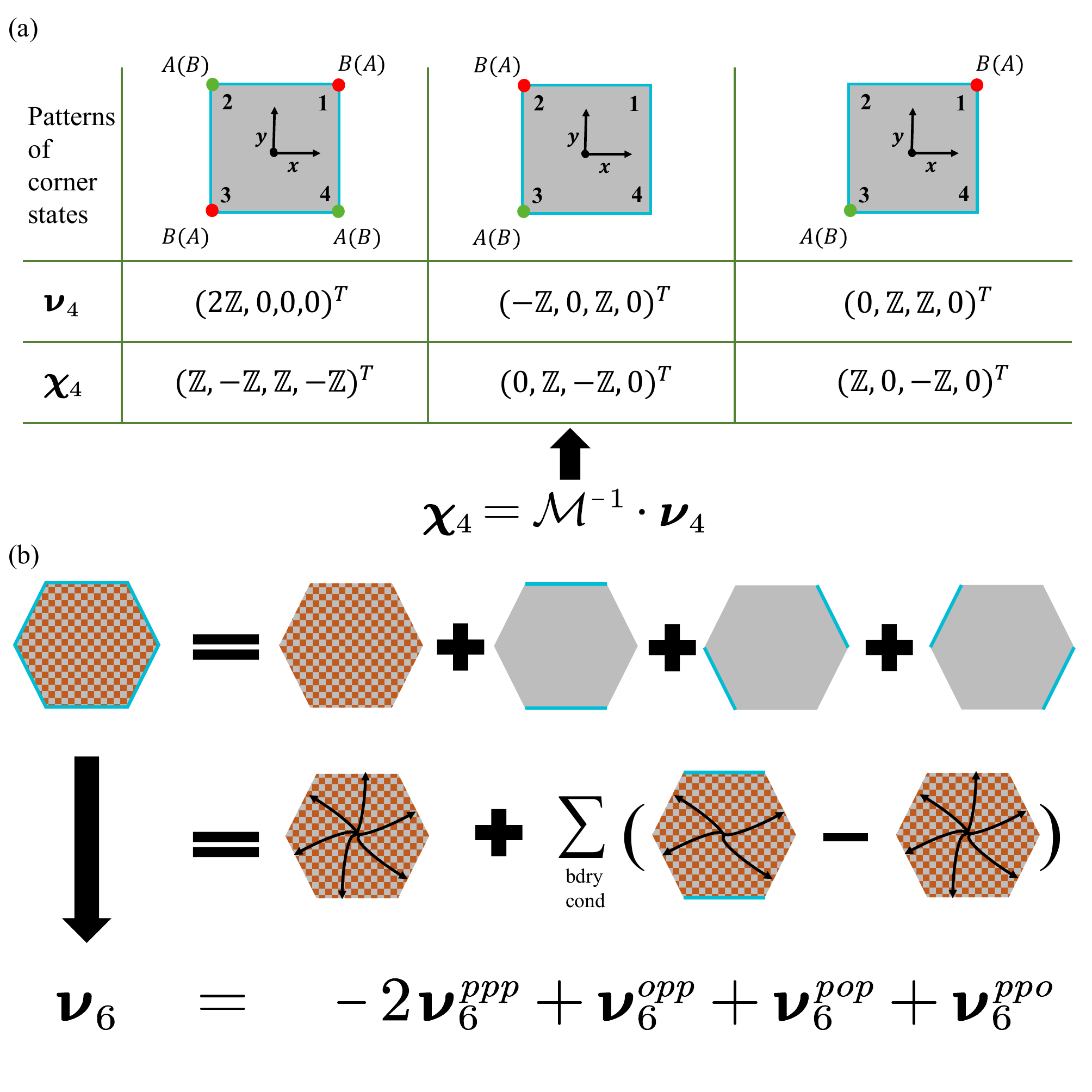}
	\caption{(a) Illustrations of three distinct corner state patterns in a square system, each characterized by a linearly independent configuration vector. The four-component Bott index vector, $\v{\nu}_4$, as defined in Eq.~\eqref{eqs:correspondence}, uniquely classifies all possible corner state configurations, with $\mathcal{M}$ given by Eq.~\eqref{eqs:M_square}.
(b) The top row shows that under open boundary conditions, the system's states consist of both bulk and boundary states. As shown in the middle row, these states can be represented as linear combinations of states from various mixed boundary conditions, where boundaries connected by black lines with double arrows are glued together. Consequently, the Bott index vector $\v{\nu}_6$ for the fully open system can be decomposed in a corresponding manner, as depicted in the bottom row. The bulk and boundary states are schematically represented by orange checker-board and cyan lines, respectively.}
	\label{fig1_configuration_sum}
\end{figure}

\medskip

\noindent{\large{\bf{Results}}}

\bigskip

\noindent{\bf{Bott index vector}}\\
%\noindent{\bf{Main results}}\\
Here, we aim to answer this question by proposing a general scheme to construct a Bott index vector, whose components are Bott indices $\nu$ that involves the polynomial $f$ of position operators and the polynomial $g$ of system length $L$,
\begin{equation*}
\begin{aligned}
    \hat{M}&=e^{2\pi i \frac{f(X,Y,Z,\dots)}{g(L)}},\\
    \nu&=\operatorname{Bott}\left(\hat{M},q\right).
\end{aligned}
\end{equation*}
For clarity, in the main text we assume that the system shape is regular (see the Supplemental information, Note 2 for the discussion of systems with non-regular shapes~\cite{hormann2006mean,floater2015generalized}). 

In chiral-symmetric systems, $n$-th order topological phases can manifest in various distinct patterns, as Fig.~\ref{fig1_configuration_sum}(a) illustrated. To provide a comprehensive description of these phases, we introduce and define the following configuration vector of corner states for system geometries possessing $m$ corners. 
\begin{equation}
\label{eqs:patterns}
    \boldsymbol{\chi}_{m}=\left(N_1^{(B)}-N_1^{(A)},\dots,N_m^{(B)}-N_m^{(A)}\right)^\mathrm{T},
\end{equation}
where $N_i^{(A)}$ and $N_i^{(B)}$ denote the number of corner states localized in the $i$-th corner and with $+$ and $-$ chirality, respectively. Given the zero trace of the chiral operator, it follows that $\sum_{i=1}^{m}{\chi}_{m}^{(i)}=0$, where ${\chi}_{m}^{(i)}$ denotes the $i$-th component of $\boldsymbol{\chi}_{m}$. 

We now introduce an universal, analytical and complete correspondence between HOTPs described by the configuration vector $\boldsymbol{\chi}_{m}$ and the topological invariants (the Bott index vector) $\v{\nu}_m$, a critical and long-standing goal in higher-order topological physics field~\cite{ExactHigherorderBulkboundary2021jung,DifficultiesOperatorbasedFormulation2019ono,TypeIIQuadrupoleTopological2020yang,HigherorderTopologicalPhases2023yanga,WannierTopologyQuadrupole2023yang}. Specifically, in the following we prove
\begin{equation}
\begin{aligned}
\label{eqs:correspondence}
  \boldsymbol{\chi}_{m}=&\mathcal{M}^{-1}\cdot \v{\nu}_{m},\\
    \v{\nu}_{m}=&\left(\nu^{(1)}_m,\dots,\nu^{(m-1)}_m,0\right)^{\mathrm{T}},
\end{aligned}
\end{equation}
where $\mathcal{M}$ is a matrix determined by $m-1$ polynomials in position operators, and the vector component $\nu^{(i)}_{m}$ represents the Bott index defined by these polynomials.

It is important to emphasize that we focus on Hamiltonians with open boundary conditions, rather than the periodic boundary conditions commonly assumed~\cite{ChiralSymmetricHigherOrderTopological2022benalcazar,ManybodyElectricMultipole2019wheeler,ManybodyOrderParameters2019kanga}. This approach ensures that $f(X,Y,Z,\dots)/g(L)$ is well-defined, which circumvents the difficulties inherent in operator-based formulations within periodic systems~\cite{DifficultiesOperatorbasedFormulation2019ono}. Moreover, it avoids issues such as the possibility of the ill-defined Bott index in cases like the multipole chiral number~\footnote{see supplemental information, Note 5.}. Additionally, we emphasize that the Bott index we define captures the topology of bulk and boundary states, thereby characterizing topologically stable ZECSs rather than directly counting corner states, as demonstrated in Eq.~\eqref{eqs:no_c_corner} and Eq.~\eqref{eqs:contri_corner}. However, adopting OBC introduces certain challenges: the $q$ is not unique when zero-energy states are present. Nevertheless, Bott index vector will remain unchanged, provided that $f$ and $g$ are properly constructed as Eq.~\eqref{eqs:Rule_1} (see details in supplemental information, Note 3.).

\medskip
\noindent\textbf{Bott index-ZECSs correspondence}\\
We begin by introducing the following theorem. 
\begin{theorem}
\label{thm:nozero}
    If no zero-energy corner states exist in a system with chiral symmetry, a finite coupling range, and an energy gap, then the following conclusion holds true in the limit of large system size $L\to \infty$.
\begin{equation*}
\begin{aligned}
&-1\notin \sigma(\hat{M}q \hat{M}^{\dagger}q^{\dagger}),\\
&\nu=\operatorname{Bott}\left(\hat{M},q\right)=0,
\end{aligned}
\end{equation*}
where $f$ and $g$ in $\hat{M}$ can be arbitrary polynomials with the same degree, i.e., $\operatorname{deg}(f)=\operatorname{deg}(g)$. Here, $\sigma$ denotes the set of eigenvalues of a matrix.
\end{theorem}

The proof of this theorem is available in Methods. Theorem~\ref{thm:nozero} establishes that a non-zero Bott index ($\nu\neq0$) under these conditions necessitates the presence of ZECSs. Furthermore, in the following we establish an analytical relationship between the Bott index $\nu$ and the configuration vector $\v{\chi}_{m}$.
\begin{equation} 
\label{eqs:bott_vector}
\nu=\sum_{i=1}^{m}\frac{\operatorname{sgn}\left(f(\boldsymbol{x}_{i}^{(c)})/g(L)\right)}{2}\chi_{m}^{(i)}\in\mathbb{Z}, 
\end{equation}
where $\boldsymbol{x}_{i}^{(c)}$ denotes the position vector of the $i$-th corner, and polynomials $f$ and $g$ satisfy the following equation
\begin{equation}
   \label{eqs:Rule_1}
    f(\boldsymbol{x}^{(c)}_{i})=\pm g(L)/2.
\end{equation}

The proof proceeds in steps. First, we decompose the Bott index based on the localization of eigenstates [\Cref{eqs:decomp,eqs:u_dec,eqs:diagonal,eqs:up_deco}].
We rewrite $\nu$ as 
\begin{equation}
\begin{aligned}
\label{eqs:decomp}
       \nu &=\frac{1}{2\pi i}\operatorname{Tr}\operatorname{log}\left(U_A^{\dagger}\hat{M}U_AU_B^{\dagger}\hat{M}^{\dagger}U_B\right)\\
        &=\frac{1}{2\pi} \sum_j \lambda_j,   
\end{aligned}
\end{equation}
where $e^{i\lambda_j}$, $\lambda_j\in(-\pi,\pi)$, is the eigenvalue of $U_A^{\dagger}\hat{M}U_AU_B^{\dagger}\hat{M}^{\dagger}U_B$~\cite{GuideBottIndex2019loring}.
We now consider $U_{\alpha}^{\dagger}\hat{M}U_{\alpha}$, where $\alpha \in \{A, B\}$. In the position basis, $U_{\alpha}$ can be written as:
\begin{equation}
\label{eqs:u_dec}
U_{\alpha}=\left[\Psi_{\alpha,\text{bulk}}\quad\Psi_{\alpha,\text{bdry}}\quad\Psi_{\alpha,\text{corner}}\right],
\end{equation}
where $\Psi_{\alpha,\text{bulk (bdry,corner)}}$ denotes the $\alpha$ component of wavefunctions extended in the bulk (boundary, corner) of the system. Given that the overlap of wavefunctions between states extended in different areas of the system decays as the system length increases and the localization property of $\hat{M}$, we obtain the following result.
\begin{equation}
\label{eqs:diagonal}
       \langle \Psi _{\alpha,\beta_1}\mid \hat{M}\mid \Psi _{\alpha,\beta_2}\rangle \stackrel{L\to \infty}{ \longrightarrow }0,\quad  \beta_1\neq \beta_2,
\end{equation}
for  $\beta_1, \beta_2 \in \{\text{bulk},\text{bdry},\text{corner}\}$. The detailed derivation is available in Methods.
Therefore, $U_{\alpha}^{\dagger}\hat{M}U_{\alpha}$ becomes a block diagonal matrix as $L\rightarrow \infty$,
\begin{equation}
\label{eqs:up_deco}
        U_{\alpha}^{\dagger}\hat{M}U_{\alpha}=\bigoplus_{\beta\in \{\text{bulk},\text{bdry},\text{corner}\}}U_{\alpha,\beta}^{\dagger}\hat{M}U_{\alpha,\beta},
\end{equation}
where $U_{\alpha,\beta}$ are matrices composed of the $\alpha$ component of eigenstates $\Psi_\beta$ reside in $\beta$.
Similarly, $U_{A}^{\dagger}\hat{M}U_{A} U_{B}^{\dagger}\hat{M}^{\dagger}U_{B}$ also becomes a block diagonal matrix.

Second, we examine the contribution of corner states to the Bott index.
We demonstrate that this contribution is identically zero [Equations \eqref{eqs:corner_as_eigen} to \eqref{eqs:contri_corner}].
This arises because, in the thermodynamic limit ($L\to \infty$), corner states $\Psi_{\alpha,\text{corner}}$ become eigenstates of the normalized position operator $\boldsymbol{X}/L$ with corresponding eigenvalues $\boldsymbol{x}_{\alpha}^{(c)}/L$, 
\begin{equation}
\label{eqs:corner_as_eigen}
    \frac{\boldsymbol{X}}{L}\Psi_{\alpha,\text{corner}}=\frac{\boldsymbol{x}_{\alpha}^{(c)}}{L}\Psi_{\alpha,\text{corner}},
\end{equation}
where $\boldsymbol{x}_{\alpha}^{(c)}$ denotes the position vector of the corner $\alpha$ where the state is localized. (A detailed derivation based on the localization property of corner states is provided in Methods section).

This eigenstate property, in conjunction with Eq.~\eqref{eqs:Rule_1}, implies that for a corner state $\Psi_{\alpha,\text{corner}}$
\begin{equation}
    \hat{M} \Psi_{\alpha,\text{corner}}=e^{2\pi i f(\v{x}_{\alpha}^{(c)})/g(L)} \Psi_{\alpha,\text{corner}}=-\Psi_{\alpha,\text{corner}}.
\end{equation}
Consequently, the operator product relevant to the Bott index calculation for these corner states simplifies to the identity matrix
\begin{equation}
\label{eqs:no_c_corner}
U_{A,\text{corner}}^{\dagger}\hat{M}U_{A,\text{corner}}U_{B,\text{corner}}^{\dagger}\hat{M}^{\dagger}U_{B,\text{corner}}=\mathbbm{1}.
\end{equation}
Given that the matrix argument of the logarithm in the Bott index formula is block-diagonal with respect to bulk, boundary, and corner states, the total Bott index can be expressed as a sum of their individual contributions
\begin{equation}
\label{eqs:no_contri_corner}
    \operatorname{Bott}\left(\hat{M},q\right)=\frac{1}{2\pi i}\sum_{\beta}\operatorname{Tr}\operatorname{log}\left(U_{A,\beta}^{\dagger}\hat{M}U_{A,\beta}U_{B,\beta}^{\dagger}\hat{M}^{\dagger}U_{B,\beta}\right).
\end{equation}
Utilizing Eq.~\eqref{eqs:no_c_corner}, the specific contribution from the corner states is therefore \begin{equation}
\label{eqs:contri_corner}
     \frac{1}{2\pi i}\operatorname{Tr}\operatorname{log}\left(\mathbbm{1}\right) = 0.
\end{equation}
This vanishing contribution from corner states signifies that any non-zero value of the Bott index $\nu$ necessarily arises from the topological properties encoded within the bulk and boundary states of the system.

Third, we consider the evolution of the remaining bulk and boundary contributions using a parameter $s$ [Equations \eqref{eqs:relation_noncorner_corner} to \eqref{eqs:final_eq}]. 

We define two functions of $s$ for $s\in[0,1]$ with $\hat{M}(s)=\operatorname{Exp}[2\pi i \frac{f(X,Y,Z,\dots)\times s}{g(L)}]$ and $\nu(s)=\operatorname{Bott}\left(\hat{M}(s),q\right)$. Thus, we have $\nu(1)=\nu$ and $\nu(s)=1/2\pi\sum_j\lambda_j(s)$. For a specific pattern of ZECSs $\v{\chi}_{m}$, we have
\begin{equation}
\begin{aligned}
\label{eqs:relation_noncorner_corner}
    &\det\left(U_{A,\text{corner}}^{\dagger}\hat{M}(s)U_{A,\text{corner}}U_{B,\text{corner}}^{\dagger}\hat{M}(s)^{\dagger}U_{B,\text{corner}}\right) \\&=e^{i\sum_{j\in\{\text{corner}\}}\lambda_j(s)}\\
    &=\det\left(\operatorname{diag}(e^{2s \pi i \frac{f\left(\v{x}_{A,i}^{(c)}\right)}{g(L)} })\operatorname{diag}(e^{-2s \pi i \frac{f\left(\v{x}_{B,i}^{(c)}\right)}{g(L)} })\right)\\
    &=e^{-i s\pi  \sum_{i}^{m}\operatorname{sgn}\left(f(\boldsymbol{x}_{i}^{(c)})/g(L)\right)\chi_{m}^{(i)}},    
\end{aligned}
\end{equation}
where in the last step we utilize Eq.~\eqref{eqs:corner_as_eigen} and `$\{\text{corner}\}$' refers to the set of indices labeling the eigenvalues of the corner block.

Given that $1=\det (\hat{M}(s)q(\hat{M}(s))^{\dagger}q^{\dagger}) = e^{i\sum_j \lambda_j(s)}$, we have
\begin{equation}
\label{eqs:corner_others}
    e^{i\sum_{j \notin \{\text{corner}\}}\lambda_j(s)}=e^{-i\sum_{j \in \{\text{corner}\}}\lambda_j(s)}.
\end{equation}
For $\lambda_{j \notin \{\text{corner}\}}(s)$, in the limit $L\to \infty$, they do not encounter the branch cut of logarithm for $s\in[0,1]$. This conclusion can be reached by applying Theorem~\ref{thm:nozero} to a effective Hamiltonian $\tilde{H}$, which is composed of boundary and bulk states and features both a spectral gap and a finite coupling range. The finite coupling range of $\tilde{H}$ is inherited from the finite coupling of $H$, as evidenced by the representation $H=\sum_{j\notin \{\text{corner}\}}E_j |\Psi_j\rangle \langle \Psi_j |$ and $\tilde{H}=R\cdot H \cdot R^{T}$, where $R$ is a rectangular matrix that projects onto the Hilbert space spanned by boundary and bulk states.
With Eq.~\eqref{eqs:relation_noncorner_corner} and Eq.~\eqref{eqs:corner_others}, we have 
\begin{equation}
    \partial_s\sum_{j\notin \{\text{corner}\}}\lambda_j(s) = \pi  \sum_{i}^{m}\operatorname{sgn}\left(f(\boldsymbol{x}_{i}^{(c)})/g(L)\right)\chi_{m}^{(i)}.
\end{equation}
Since $\lambda_{j}(0)=0$, after integrating both sides, we have $\sum_{j\notin \{\text{corner}\}}\lambda_{j}(s)= s\pi  \sum_{i}^{m}\operatorname{sgn}\left(f(\boldsymbol{x}_{i}^{(c)})/g(L)\right)\chi_{m}^{(i)}$. When $s=1$, we obtain
\begin{equation}
\begin{aligned}
\label{eqs:final_eq}
       \nu=\nu(1)&=\frac{1}{2\pi } \sum_j \lambda_j(1)=\frac{1}{2\pi } \sum_{j \notin\{\text{corner}\}} \lambda_{j}(1)\\&=\sum_{i=1}^{m}\frac{\operatorname{sgn}\left(f(\boldsymbol{x}_{i}^{(c)})/g(L)\right)}{2}\chi_{m}^{(i)},
\end{aligned}
\end{equation}
where in the first row we use the result obtained in Eq.~\eqref{eqs:no_c_corner}, which states that $\lambda_j(1)=0$ for all $j\in\{\text{corner}\}$.

\medskip
\noindent\textbf{Full characterization of systems with arbitrary shape}\\
Noticing that Equation~\eqref{eqs:final_eq} is a linear equation and $\sum_{i=1}^{m}\chi_{m}^{(i)}=0$, we are able to establish a one to one map for the configuration vector of HOTPs and an $m$-component Bott index vector $\v{\nu}_{m}$. The first $m-1$ components of $\v{\nu}_{m}$, $\nu^{(i)}_{m}$, are the Bott indices defined using $m-1$ distinct sets of polynomials, $f_{m}^{(i)}$ and $g_{m}^{(i)}$, which satisfy the condition in Eq.~\eqref{eqs:Rule_1}. The $m$-th component of $\v{\nu}_{m}$ is set to zero, as a consequence of $\sum_{i=1}^{m}\chi_{m}^{(i)}=0$. Next, define the $m\times m$ matrix $\mathcal{M}$ with $\mathcal{M}_{ij}=\operatorname{sgn}\left(f_{m}^{(i)}(\boldsymbol{x}_{j}^{c})/g_{m}^{(i)}(L)\right)/2$ for $1\le i \le m-1$ and $\mathcal{M}_{mj}=1/2$ (where $j$ labels the corners). By choosing the polynomials such that $\det \mathcal{M} \neq 0$, we establish an analytical correspondence, namely Eq.~\eqref{eqs:correspondence}.
\begin{equation*}
     \boldsymbol{\chi}_{m}=\mathcal{M}^{-1}\cdot \left(\nu^{(1)}_m,\dots,\nu^{(m-1)}_m,0\right)^{\mathrm{T}}.
\end{equation*}
By implementing this equation, we are able to characterize all $n$-th order chiral symmetry protected topological phases for systems with $m$ corners.

Taking systems with the square shape as an example, we define the Bott index vector $\v{\nu}_4$ with each component corresponding to $f/g$ being $2XY/L^2$, $X/L$, and $Y/L$, respectively, where $L$ is the side length of the square. As illustrated in Fig.~\ref{fig1_configuration_sum}(a), we characterize patterns using Eq.~\eqref{eqs:correspondence}, where
\begin{equation}
\label{eqs:M_square}
\mathcal{M}=\frac{1}{2}\begin{pmatrix}
    1 & -1 & 1 &-1 \\
    1 & -1 & -1 &1 \\
    1 & 1 & -1 &-1 \\
    1 & 1 & 1 &1
\end{pmatrix}.
\end{equation}
For systems shaped as regular hexagons, 
we define the Bott index vector $\v{\nu}_6$ with the corresponding $f/g$,
\begin{equation}
\label{eqs:polyc6}
\begin{aligned}
     \{&(X^3-\frac{XY^2}{3}+\frac{8Y^3}{3\sqrt{3}})/2L^3, (X^2-\frac{4XY}{\sqrt{3}}-\frac{Y^2}{3})/2L^2,\\ &(X^2+\frac{4XY}{\sqrt{3}}-\frac{Y^2}{3})/2L^2,(X^3-3XY^2)/2L^3,\\ &(X^3+\frac{7XY^2}{3})/2L^3\},
\end{aligned}
\end{equation}
where $L$ is again the side length of the hexagon. 

Those polynomials arise from the nontrivial representations of the cyclic group corresponding to the rotational symmetries of a regular polygon. Together with the trivial polynomials $1/2$, they form a faithful representation of the cyclic group $\mathrm{C}_{m}$. In the Method section, we provide an algorithm for generating them with the corresponding $\mathcal{M}\neq 0$. For reference, in Supplemental information, Note 4, we provide a series of polynomials and their associated matrices $\mathcal{M}$ for systems shaped as $l$-sided regular polygons, where $l \in \{4, 5, 6, 8, 12\}$. Systems with non-regular shapes are discussed in Supplemental information, Note 2.

\medskip
\noindent\textbf{Sum rules under different boundary conditions}\\
To connect our OBC framework with conventional PBC approaches and understand the origin of the Bott index vector, we examine how $\v{\nu}_m$ decomposes under different boundary conditions. We analyze the behavior of Bott index vector by considering $U_{\alpha}$ under varying boundary conditions, denoted as $U_{\alpha}^{s_1 \dots s_k \dots s_y}$. Here, the superscript indicates the boundary condition for each of the $y$ independent boundaries, which can be either periodic ($s_k = \text{p}$) or open ($s_k = \text{o}$). Our result demonstrates that the Bott index vector $\v{\nu}_m$, calculated for a system with all open boundaries, can be expressed as a sum of Bott index vectors obtained with partially periodic boundary conditions. The partial periodic boundary conditions here are designed to account for how the boundaries are ``connected" or ``glued" together, rather than simply being periodic in the conventional sense (see Methods for details on boundary conditions).

In the thermodynamic limit, as $L \to \infty$, the unitary matrices $U_{\alpha}^{s_1 \dots s_k \dots s_y}$ exhibit a hierarchical structure. Systems with exclusively periodic boundaries comprise only bulk states. Conversely, each introduced open boundary contributes an additional, distinct set of boundary-localized states.
\begin{equation}
\begin{aligned}
\label{eqs: u_structure}
    U_{\alpha}^{p\dots p}&=\left[\Psi_{\alpha,\text{bulk}}\right], \\
    U_{\alpha}^{op\dots p} &=\left[\Psi_{\alpha,\text{bulk}}\quad\Psi_{\alpha,\text{bdry\textunderscore 1}}\right],\\
    U_{\alpha}^{po\dots p} &=\left[\Psi_{\alpha,\text{bulk}}\quad\Psi_{\alpha,\text{bdry\textunderscore 2}}\right], \\
    \vdots & \\
    U_{\alpha}^{o\dots o} &=\left[\Psi_{\alpha,\text{bulk}}\quad\Psi_{\alpha,\text{bdry\textunderscore 1}} \cdots \right],
\end{aligned}
\end{equation}

This hierarchical arrangement of states implies an additive composition for the Bott index vector. With prior observations, Eq.~\eqref{eqs:no_contri_corner}, the Bott index for a system with one open boundary (e.g., along the first direction) can be expressed as the sum of the bulk contribution (all boundaries periodic, $\nu_m^{(i),p\dots p}$) and the contribution from states localized at that specific boundary $\nu_{m,\text{bdry\textunderscore 1}}^{(i)}$.
\begin{equation}
        \nu_{m}^{(i),op\dots p} =\nu_m^{(i),p\dots p}+\nu_{m,\text{bdry\textunderscore 1}}^{(i)},
\end{equation}
Analogous expressions can be derived for other mixed boundary configurations.

This framework allows for the systematic isolation of Bott index contributions stemming from boundary features of different dimensionalities. We define $\nu^{(i)}_{m,\text{j-bd}}$ (for $j\le d-1$) as the net contribution to Bott index from all $(n-j)$-dimensional boundary states ($j=0$ represents bulk states). This term can be extracted via a recursive procedure. The sum of Bott indices over all $\binom{y}{j}$ (the number of combination) configurations with exactly $j$ open boundaries incorporates contributions from all boundary features of dimension $(n-l)$ where $l \le j$. To isolate $\nu^{(i)}_{m,\text{j-bd}}$, the contributions from higher-dimensional boundaries ($l < j$) must be subtracted. A specific $(n-l)$-dimensional boundary feature is overcounted, appearing in all configurations where its $l$ defining boundaries are open, plus any $j-l$ of the remaining $y-l$ boundaries. The overcounting factor is, therefore, $\binom{y-l}{j-l}$. This logic yields the recursive relation:
\begin{equation}
\begin{aligned}
        \nu^{(i)}_{m,\text{0-bd}} &\coloneqq \nu_{m}^{(i),p\dots p},\\
        \nu^{(i)}_{m,\text{1-bd}} &\coloneqq \sum_{k\in S_{1,y}} \nu_{m}^{(i),k}-\binom  {y} {1}  \nu^{(i)}_{m,\text{0-bd}},\\
        & \vdots \\
        \nu^{(i)}_{m,\text{j-bd}} &\coloneqq \sum_{k\in S_{j,y}} \nu_{m}^{(i),k}-\sum_{l=0}^{j-1}\binom {y-l} {j-l}  \nu^{(i)}_{m,l\text{-bd}}.
\end{aligned}
\end{equation}
$S_{j,y}$ is the set of all distinct permutations, where each permutation $k$ is formed from a sequence containing $j$ `o' elements and $y-j$ `p' elements.

\begin{figure}[t]
    \centering
    \includegraphics[width=1.0\columnwidth]{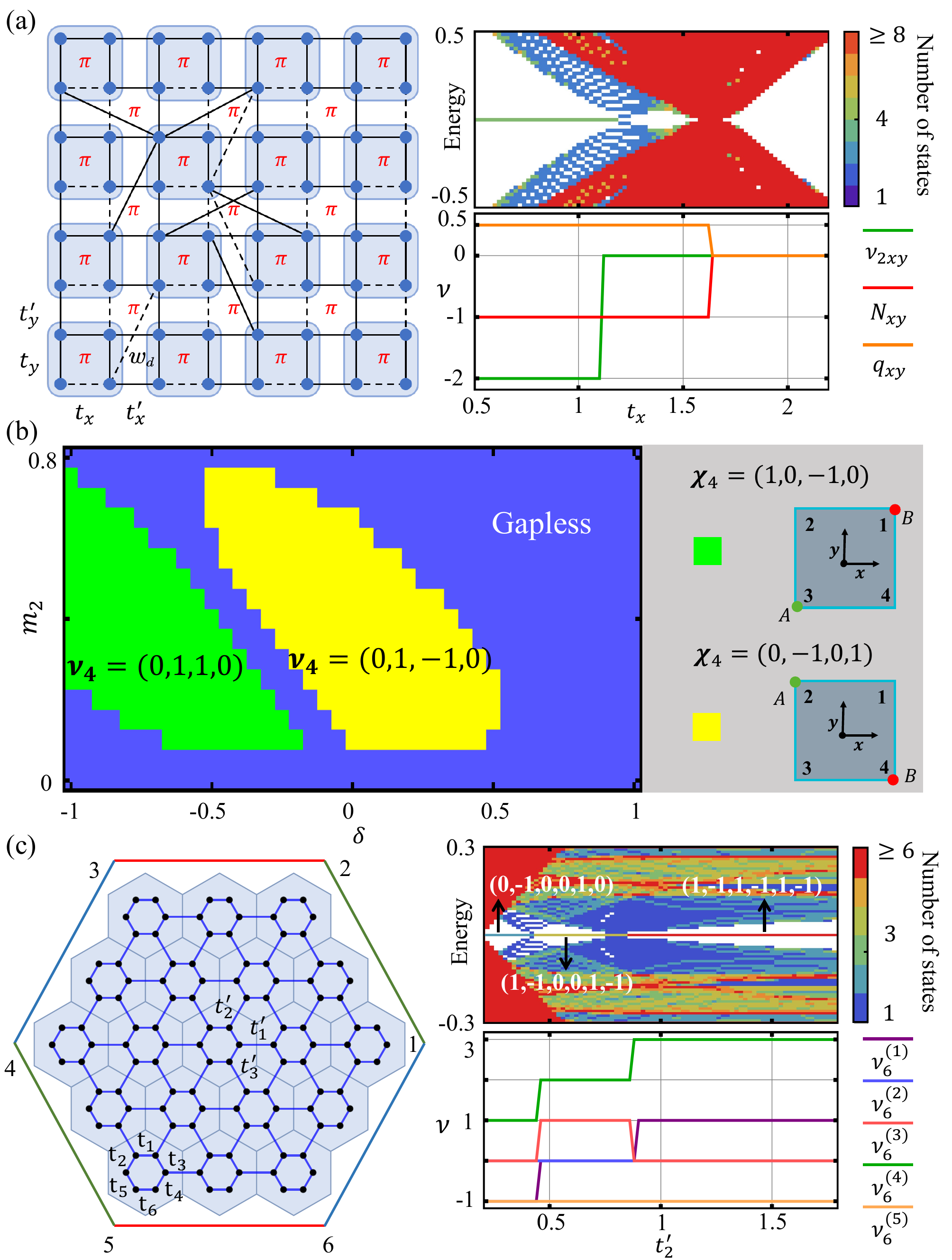}
    \caption{(a) Left panel: Schematic of the model in Eq.~\eqref{eqs:hami_without}. For clarity, we do not show all non-nearest-neighbor hoppings. The dashed lines represent an additional phase factor of $-1$ such that the system possesses the staggered $\pi$ flux. Right panel: Density of states for this system and corresponding $\nu_{2xy}$, $q_{xy}$ and $N_{xy}$ as functions of $t_x$ with $t_y=0.1$, $t_x^{\prime}=t_y^{\prime}$, $w_d=0.8$. $\nu_x$ and $\nu_y$ always equal $0$. The system length $L=50$. (b) Left panel: Phase diagram of the model governed by Eq.~\eqref{eqs:hami_mirr} as a function of $\delta$ and $m_2$ with $m_1=1$. The topological invariant $\v{\nu}_4$ is used to identify three distinct phases, colored green, yellow, and blue. Right panel: Illustration of how the green and yellow phases correspond to different patterns of zero-energy corner states (ZECSs) and their associated configuration vectors, $\v{\chi}_4$. (c) Left panel: Schematic of the model in Eq.~\eqref{eqs:hami_c6}. Right panel: The upper part shows the density of states for this system and the configuration vectors $\v{\chi}_{6}$ for three higher-order topological phases. The lower part displays the corresponding Bott indices $\nu^{(1,\cdots,5)}_6$ as functions of $t_2^{\prime}$. We fix $t_1=1/10$, $t_2=1/3$, $t_3=1/4$, $t_4=3/10$, $t_5=1/4$, $t_6=1/2$, $t_1^{\prime}=t_3^{\prime}=1$ and the system length $L=30$.}
    \label{fig:eff_mir_c6}
\end{figure}

Solving this recursive system yields a general expression for $\nu^{(i)}_{m,\text{j-bd}}$ in terms of the Bott index calculated for various mixed boundary conditions, we have
\begin{equation}
    \nu^{(i)}_{m,\text{j-bd}} = \sum_{l=0}^{j}(-1)^{j-l} \binom {y-l} {j-l}\sum_{k\in S_{l,y}} \nu_{m}^{(i),k}.
\end{equation}

This decomposition culminates in a general sum rule for the Bott index vector $\v{\nu}_m$. By relating the index to the individual boundary contributions, $\nu^{(i)}_m=\sum_{j=0}^{n-1} \nu^{(i)}_{m,\text{j-bd}}$, we arrive at the expression.
\begin{equation}
\label{eqs:sum_rule}
    \v{\nu}_m=\sum_{l=0}^{n-1}(-1)^{n-1-l} \binom {y-1-l} {n-1-l}\sum_{k\in S_{l,y}} \v{\nu}_{m}^{k}.
\end{equation}

The general expression in Eq.~\eqref{eqs:sum_rule} is best understood through specific examples. For one-dimensional systems, this simplifies to $\nu=\nu^p$, consistent with Eq.~\eqref{eqs:clasc_bott}. In two dimensions, we obtain
\begin{align}
    \v{\nu}_4 &=-\v{\nu}_4^{pp}+\v{\nu}^{op}_4+\v{\nu}^{po}_4, \label{eqs:sum_rules_c4}\\
        \v{\nu}_6 &=-2\v{\nu}_6^{ppp}+\v{\nu}^{opp}_6+\v{\nu}^{pop}_6+\v{\nu}^{ppo}_6.\label{eqs:sum_rules_c6}
\end{align}
A schematic derivation for the hexagonal case, as described by Eq.~\eqref{eqs:sum_rules_c6}, is illustrated in Fig.~\ref{fig1_configuration_sum}(b), which shows how the simple decomposition of states under open boundary conditions results in the sum rule.
While  Eq.~\eqref{eqs:sum_rules_c4} for $\v{\nu}_4$ of the rectangle shape systems bears formal resemblance to the multipole decompositions discussed in Ref.~\cite{ElectricMultipoleMoments2017benalcazar} (i.e., $Q^{\text{corner}} = -q_{xy} + p_x^{\text{edge}} + p_y^{\text{edge}}$), the physical interpretation differs fundamentally. Our sum rule decomposes a global topological index based on contributions from bulk and boundary states under varying boundary configurations; thus it depends on system shapes, whereas multipole approaches typically sum local moments or polarizations. This distinction is further clarified in three-dimensional systems, where our framework predicts that for a cube we have
\begin{equation}
\v{\nu}_8=\v{\nu}_8^{ppp}+\v{\nu}_8^{poo}+\v{\nu}_8^{opo}+\v{\nu}_8^{oop}-\v{\nu}_8^{opp}-\v{\nu}_8^{pop}-\v{\nu}_8^{ppo},
\end{equation}
which contrasts with earlier 3D multipole formulations. This decomposition emphasizes how our framework captures the interplay between interior and boundary properties, distinguishing it from prior multipole analyses.

\medskip
\noindent\textbf{Lattice Models}\\ We now provide three concrete examples to demonstrate the application of our framework. We note that all calculations of Bott index vectors are performed in real space. 

First, we examine a modified model, informed by two recent papers \cite{QuadrupoleInsulatorCorner2023tao,PhysRevLett.132.213801}, which incorporates additional diagonal long-range hoppings. These hoppings break separability while maintaining momentum-glide reflection symmetries~\cite{PhysRevLett.130.256601,chen2022brillouin} and chiral symmetry, as depicted in the left panel of Fig.~\ref{fig:eff_mir_c6}(a). The corresponding Bloch Hamiltonian at momentum $\v{k}$ can be written in the form of Eq.~\eqref{eqs:chiralhami} with
\begin{equation}
    \label{eqs:hami_without}
    \begin{aligned}
        h(\v{k})=&\left( \begin{matrix}
	t_x+t_x^{\prime} e^{-i k_x}&		t_y+t_y^{\prime} e^{-i k_y}\\
	t_y-t_y^{\prime} e^{i k_y}&		-t_x+t_x^{\prime} e^{i k_x}\\
\end{matrix} \right)\\&+w_d\left( \begin{matrix}
	i e^{-i k_x}\cos k_y&		i e^{-i k_y}\cos k_x\\
	i e^{i k_y}\cos k_x&		-i e^{i k_x}\cos k_y\\
\end{matrix} \right),    
    \end{aligned}
\end{equation}
where $t_{x(y)}$, $t_{x(y)}^{\prime}$, and $w_d$ are the nearest hoppings within unit cells, between unit cells, and along the diagonal directions, respectively.
We set $t_x^{\prime}=t_y^{\prime}=1$. This system hosts one ZECS at each corner, which cannot be accurately characterized by the quadrupole moment $q_{xy}$~\cite{ElectricMultipoleMoments2017benalcazar,TopologicalPhaseTransitions2020li,ManybodyElectricMultipole2019wheeler,ManybodyOrderParameters2019kanga} or the multipole chiral number $N_{xy}$ proposed in Ref.~\cite{ChiralSymmetricHigherOrderTopological2022benalcazar} (definitions provided in Supplemental information, Note 5), mirroring findings reported in the aforementioned two papers~\cite{QuadrupoleInsulatorCorner2023tao,PhysRevLett.132.213801}. However, the Bott index vector $\v{\nu}_4=(\nu_{2xy},\nu_x,\nu_y)$ is in precise agreement with the existence of ZECSs, unlike $q_{xy}$ and $N_{xy}$, as depicted by the density of states as a function of $t_x$ in Fig.~\ref{fig:eff_mir_c6}(a). Since the calculation shows that Bott indices $\nu_{x}$ and $\nu_{y}$ always equal zero, we only present $\nu_{2xy}$ as a function of $t_x$, which is consistent with the ZECS pattern [see Fig.~\ref{fig1_configuration_sum}(a)].

Second, we study another system with ZECSs located only at two diagonal corners of a square. The corresponding model has mirror-symmetry-protected corner states, as reported by Refs.~\cite{SecondorderTopologicalInsulators2018geier, PhysRevLett.124.166804}, and we add a mirror-symmetry-breaking $\delta$ term to it.
\begin{equation}
\label{eqs:hami_mirr}
\begin{aligned}
       h(\v{k})= &(\delta+m_2-i \sin k_y)\mathbbm{1}+(-\sin k_x+i m_2)\tau_z \\&+(-m_1-\cos k_x-\cos k_y)\tau_x,   
\end{aligned}
\end{equation}
where $\tau_{i}$ represents the Pauli matrix. When $\delta=0$, this system possesses diagonal and anti-diagonal mirror symmetries. The HOTPs in this system, which appear even without mirror symmetries ($\delta\neq0$), can be comprehensively characterized by the Bott index vector $\v{\nu}_4$. The phase diagram of $\v{\nu}_4$ as a function of $\delta$ and $m_2$ is shown in Fig.~\ref{fig:eff_mir_c6}(b), illustrating that two HOTPs, where two ZECSs are located at diagonal or anti-diagonal corners, are separated by gapless edge phases. The relationship between the Bott index vector $\v{\nu}_4$ of each phase and the ZECSs configuration vector $\v{\chi}_{4}$ is described by Eq.~\eqref{eqs:correspondence} with $\mathcal{M}$ in Eq.~\eqref{eqs:M_square}, as illustrated in the right panel of Fig.~\ref{fig:eff_mir_c6}(b). 

Finally, we study a lattice model with a regular hexagon shape, as shown in the left panel of Fig.~\ref{fig:eff_mir_c6}(c):
\begin{equation}
\label{eqs:hami_c6}
       h(\v{k})= \begin{pmatrix}
	t_1 &		t_2 & t_3^{\prime} e^{\frac{i(k_x-\sqrt{3}k_y)}{2}}\\
	t_1^{\prime} e^{-i k_x}&	t_3 &t_4\\
 t_5 & t_2^{\prime} e^{\frac{i(k_x+\sqrt{3}k_y)}{2}} &t_6
\end{pmatrix}.
\end{equation}
When $t_i=t$ and $t_j^{\prime}=t^{\prime}$ for $i\in\{1,\dots,6\},j\in\{1,2,3\}$, the system possesses $C_6$ symmetry with its ZECSs characterized by the $C_6$ topological indices~\cite{QuantizationFractionalCorner2019benalcazar}. However, in the absence of the $C_6$ and $C_2$ symmetries, the ZECSs still emerge, and HOTPs in this system exhibit a varying number of ZECSs. This is illustrated by the density of states as a function of $t_2^{\prime}$, as shown in the right panel of Fig.~\ref{fig:eff_mir_c6}(c). As $t_2^{\prime}$ increases from $0.3$ to $1.8$ (with other parameters specified in Fig.~\ref{fig:eff_mir_c6}), the system changes from a HOTP with two ZECSs to one with four ZECSs, and subsequently to one with six. Each phase transition is marked by the closure of the edge gap. In Fig.~\ref{fig:eff_mir_c6}(c), we also show the Bott index vector $\v{\nu}_6$ as functions of $t_2^{\prime}$, which fully characterizes each phase according to Eq.~\eqref{eqs:correspondence} with $\mathcal{M}$ provided in the Supplemental information, Note 4. Furthermore, $\v{\nu}_6$ was independently calculated using the sum rule stated in Eq.~\eqref{eqs:sum_rules_c6}, and the agreement between these two determinations of $\v{\nu}_6$ validates the sum rule Eq.~\eqref{eqs:sum_rule}.

\bigskip
\noindent{\large{\bf{Discussion}}}\\
In summary, we have established a comprehensive framework to universally characterize $n$-th order topological phases in $n$-dimensional chiral-symmetric systems, which can be intrinsic or extrinsic HOTPs depending on the presence of certain crystalline symmetries (for the discussion on crystalline symmetries and additional model demonstrations, see the Supplemental information, Note 6 and 7, respectively). First, we have established an exact correspondence between the Bott index and the HOTPs, addressing the longstanding challenges reported in earlier studies~\cite{QuadrupoleInsulatorCorner2023tao,PhysRevLett.132.213801,PhysRevB.107.045118,TypeIIQuadrupoleTopological2020yang,10.21468/SciPostPhys.17.2.060,PhysRevResearch.2.043012}. Second, by providing a general strategy to construct Bott index vectors, our framework applies universally to $n$-dimensional systems of arbitrary shapes. Furthermore, our framework enables the characterization of all possible patterns of HOTPs, revealing their spatial topological information using Bott index vectors. Even in situations where corner states shift from the geometric corners~\cite{zhang2020all}, our framework remains robust as long as the displacement is smaller than the order of $L$. We anticipate that our theory will find broad applications in the characterization and design of HOTPs, including, for example, topological superconductors with Majorana corner modes~\cite{luo2024spinbottindicestimereversalinvariant}, and topological polariton corner state lasing~\cite{wu2023higher}. Serving as both a theoretical foundation and a powerful practical tool, this framework opens new avenues for exploring and implementing higher-order topological phases.

%\newpage
%~\newpage
\bigskip
\noindent \textbf{\large METHODS}\\
\noindent \textbf{Notations and Inequalities}\\
$\sigma(\cdot)$ denotes the set of eigenvalues of a matrix.

$\underset{S}{\text{sup}} P$ represents the supremum of the values taken by $P$ over a set $S$. 

$\parallel\cdot\parallel$ denotes the spectral norm of a matrix (the largest singular value of a matrix). This norm is induced by the Euclidean norm, $|\cdot|$, for vectors and is given by $\parallel A\parallel=\underset{x\neq0}{\text{sup}}\frac{|Ax|}{|x|}$, where $x$ is a vector.

$\operatorname{dist}(n,m)$ represents the Euclidean distance function in position space.

$\mathcal{O}$ denotes the order of approximation.

For the spectral norm, we have following two inequalities for two square matrices $A$ and $B$,
\begin{eqnarray*}
    \parallel A + B \parallel &\le \parallel A \parallel + \parallel B \parallel \\
    \parallel A  B \parallel &\le \parallel A \parallel  \parallel B \parallel.
\end{eqnarray*}

\bigskip
\noindent \textbf{Mathematical Structure of Bott Index}

\begin{sectiondefinition}
\label{def:Bott index}
Given two unitary matrices $U$ and $V$, such that ${-1}\notin \sigma(UVU^{\dagger}V^{\dagger})$ or equivalently such that $\parallel [U,V]\parallel <2$, their Bott index is defined as:
\begin{equation}
    \operatorname{Bott}(U,V):=\frac{1}{2\pi i} \operatorname{Tr} \operatorname{log}\left(UVU^{\dagger}V^{\dagger}\right).
\end{equation}
\end{sectiondefinition}
\begin{sectionremark}
\label{rem:equi}
    From 
    \begin{equation}
        \parallel UVU^{\dagger}V^{\dagger}-\mathbbm{1}\parallel=\parallel (UV-VU)U^{\dagger}V^{\dagger}\parallel=\parallel[U,V]\parallel,
    \end{equation}
    we deduce that ${-1}\in \sigma(UVU^{\dagger}V^{\dagger})$ if and only if $\parallel[U,V]\parallel=2$. Given the unitarity of $U$ and $V$, we have $\parallel[U,V]\parallel=\parallel (UV-VU)\parallel\le \parallel UV \parallel+\parallel VU\parallel=2$. Thus, since ${-1}\notin \sigma(UVU^{\dagger}V^{\dagger})$ is equivalent to $\parallel[U,V]\parallel\neq2$, we deduce that ${-1}\notin \sigma(UVU^{\dagger}V^{\dagger})$ is also equivalent to $\parallel[U,V]\parallel<2$.
\end{sectionremark}
The Bott index in Definition~\ref{def:Bott index} is an Integer. This can be obtained by the following~\cite{BottIndexTwo2021toniolo,BottIndexUnitary2022toniolo,GuideBottIndex2019loring}. From $\det (UVU^{\dagger}V^{\dagger})=1$, we have $1=\prod_j e^{i \theta_j}=e^{i\sum_j \theta_j}$, where $e^{i \theta_j}$ with $\theta_j \in (-\pi,\pi)$, is the eigenvalue of $UVU^{\dagger}V^{\dagger}$. It follows that $\operatorname{Bott}(U,V)=\frac{1}{2\pi i}\sum_j \operatorname{log} (e^{i \theta_j})=\frac{1}{2\pi}\sum_j\theta_j \in \mathbb{Z}$.
\begin{sectiontheorem}
    \label{thm: homoinvar}
    Given two continuous maps $V(s)$ : $ [0,1]\to \mathcal{U}(N)$ and $W(s)$ : $ [0,1]\to \mathcal{U}(N)$, where $\mathcal{U}(N)$ represents the unitary group, with $V(0)=V$, $W(0)=W$, such that $\parallel [V(s),W(s)]\parallel <2,\forall s\in [0,1]$, then
    \begin{equation}
        \label{eqs: homoinvar}
        \operatorname{Bott}\left(V(s),W(s)\right)=\operatorname{Bott}(V,W).
    \end{equation}
\end{sectiontheorem}
\begin{proof}
    The proof of this theorem can be found in Refs.~\cite{BottIndexTwo2021toniolo,BottIndexUnitary2022toniolo}.
\end{proof}
As demonstrated by this theorem, Bott index is a topological invariant. A continuous transformation that keeps it well-defined cannot alter the Bott index. With this theorem, we have the following corollary.
\begin{sectioncorollary}
\label{coro: existence}
Consider two unitary matrices $V=e^{i A}$ and $W=e^{i B}$, where $A$ and $B$ are Hermitian matrices. If $\operatorname{Bott}(V,W)\neq0$, then there exists $s_0\in [0,1]$ such that $\parallel[e^{i A s_0},e^{i B}] \parallel=2$. 
\end{sectioncorollary}
\begin{proof}
    Consider the case where $\parallel[e^{i A s_0},e^{i B}] \parallel<2$ for all $s_0\in [0,1]$. Consequently, according to Theorem~\ref{thm: homoinvar}, $\operatorname{Bott}\left(e^{i A},e^{i B}\right)=\operatorname{Bott}\left(e^{i A 0},e^{i B}\right)=0$. However, since $\operatorname{Bott}\left(e^{i A},e^{i B}\right)\neq 0$, this indicates the existence of at least one value $s_0\in[0,1]$ for which $\parallel[e^{i A s_0},e^{i B}] \parallel=2$.
\end{proof}

Consider an $n$-dimensional Hamiltonian $H$ within a lattice of length $L$ and coupling range $R \ll L$, under OBC. This implies $H_{n,m}=0$ when $\operatorname{dist}(n,m)>R$ in the position basis. We construct position operators $X,Y,Z,\dots$ according to the dimension of this Hamiltonian. We have the following theorem for $H$ under OBC.
\begin{sectiontheorem}
    \label{thm: hamiNorm}
    Consider two polynomials $f(X,Y,Z,\dots)$ and $g(L)$ with $\operatorname{deg}(f)=\operatorname{deg}(g)$. Then, 
    \begin{equation}
        \parallel[e^{2\pi i \frac{f(X,Y,Z,\dots)}{g(L)}},H]\parallel\le \mathcal{O}\left(\frac{R}{L}\parallel H\parallel \right),
    \end{equation}
    for $H$ under open boundary conditions.
\end{sectiontheorem}
This proof is based on ideas from Ref.~\cite{BottIndexUnitary2022toniolo}.
\begin{proof}
     We use the Holmgren bound~\cite{FunctionalAnalysis2002lax,BottIndexUnitary2022toniolo} for the norm of a bounded operator $A$.
    \begin{equation}
        \parallel A\parallel\le \max \left(\sup_{m\in \mathbb{Z}^n}\sum_{n\in \mathbb{Z}^n} |\langle m|A|n\rangle|, m\leftrightarrow n\right),
    \end{equation}
    where $|n\rangle$ and $|m\rangle$ are eigenkets of position operator, and $m\leftrightarrow n$ denotes the exchange of the index $m$ and $n$ in the supremum and in the sum. A proof of this bound can be found in Ref.~\cite{FunctionalAnalysis2002lax} and Ref.~\cite{BottIndexUnitary2022toniolo}. Thus, we have
    \begin{equation}
    \begin{aligned}
                &\parallel [e^{2\pi i \frac{f(X,Y,Z,\dots)}{g(L)}},H]\parallel \\& \le \max \left(\sup_{m\in \mathbb{Z}^n}\sum_{n\in \mathbb{Z}^n} |\langle m|[e^{2\pi i \frac{f(X,Y,Z,\dots)}{g(L)}},H]|n\rangle|, m\leftrightarrow n\right).
    \end{aligned}
    \end{equation}
    Notice that
    \begin{equation}
        \langle m|[e^{2\pi i \frac{f(X,Y,Z,\dots)}{g(L)}},H]|n\rangle=(e^{2\pi i \frac{f(\boldsymbol{x}_m)}{g(L)}}-e^{2\pi i \frac{f(\boldsymbol{x}_n)}{g(L)}})\langle m|H|n\rangle,
    \end{equation}
    where $\boldsymbol{x}_m$ and $\boldsymbol{x}_n$ are the coordinate of $|m\rangle$ and $|n\rangle$, respectively. Therefore, 
\begin{align*}
&\|\bigl[e^{2\pi i \frac{f(X,Y,Z,\dots)}{g(L)}},H\bigr]\| \\
&\le \max_{m\leftrightarrow n} \Biggl(\sup_{m\in \mathbb{Z}^n}\sum_{\substack{\text{dist}(n,m)\le R}} \bigl|e^{2\pi i \frac{f(\boldsymbol{x}_m)}{g(L)}}-e^{2\pi i \frac{f(\boldsymbol{x}_n)}{g(L)}}\bigr||\langle m|H|n\rangle|\Biggr).
\end{align*}
Noting that for $\text{dist}(n,m)\le R$, we have the following relationship hold
    \begin{equation}
    \begin{aligned}
    \label{eqs:le_order}
       |e^{2\pi i \frac{f(\boldsymbol{x}_m)}{g(L)}}-e^{2\pi i \frac{f(\boldsymbol{x}_n)}{g(L)}}| &=|e^{2\pi i \frac{f(\boldsymbol{x}_m)}{g(L)}}(1-e^{2\pi i \frac{f(\boldsymbol{x}_n)-f(\boldsymbol{x}_m)}{g(L)}})|\\&\le\mathcal{O}(\frac{R}{L}), 
    \end{aligned}
    \end{equation}
    where we use $\frac{f(\boldsymbol{x}_n)-f(\boldsymbol{x}_m)}{g(L)}\le\mathcal{O}\left(\frac{R}{L}\right)$. This is obtained by considering the Taylor expansion of $f$ at $\boldsymbol{x}_m$
    \begin{equation}
    \begin{aligned}
        \frac{f(\v{x}_n)}{g(L)} &=\frac{f(\v{x}_m)+\nabla  f(\v{x})|_{\v{x}=\v{x}_m}\cdot(\v{x}_n-\v{x}_m)+\cdots}{g(L)}\\
        &\le\frac{f(\v{x}_m)}{g(L)}+\mathcal{O}(R/L)+\mathcal{O}(R/L^2)+\cdots,
    \end{aligned}
    \end{equation}
    where we use $\operatorname{deg}(f)=\operatorname{deg}(g)$.
    It follows that
    $$
        \parallel[e^{2\pi i \frac{f(X,Y,Z,\dots)}{g(L)}},H]\parallel\le \mathcal{O}\left(\frac{R}{L}\parallel H\parallel \right).
    $$
\end{proof}
\begin{sectionlemma}
\label{lem: normP}
    Given a Hamiltonian $H$, as defined above, with a spectral gap $\Delta E$, the spectral projector $P$ is called the Fermi projection, defined as follows:
    \begin{equation}
        P=\sum_{i=1}^{N_{occ}}|\Psi_i\rangle \langle \Psi_i |.
    \end{equation}
    The summation is taken over all occupied states. The following relationship holds true,
    \begin{equation}
        \parallel[e^{2\pi i \frac{f(X,Y,Z,\dots)}{g(L)}},P]\parallel\le \mathcal{O}\left(\frac{R}{L} \frac{\parallel H\parallel}{\Delta E} \right),
    \end{equation}
    for arbitrary polynomials $f$ and $g$ with $\operatorname{deg}(f)=\operatorname{deg}(g)$.
\end{sectionlemma}
The proof of this Lemma is based on an idea from Ref.~\cite{BottIndexUnitary2022toniolo}.
\begin{proof}
 For $z \notin \sigma(H)$ and any matrix $A$ with the size of $H$, we have an equality
    \begin{equation}
    \label{eqs:commutor_equa}
        \begin{aligned}
\left[A,(H-z)^{-1}\right] &=(H-z)^{-1}\left[(H-z),A\right](H-z)^{-1}  \\ &=(H-z)^{-1}\left[H,A\right](H-z)^{-1}.
\end{aligned}
    \end{equation}
Now, we use the contour integral representation of the Fermi projection $P$, with the loop $\Gamma$ in the complex plane enclosing the eigenvalues below the Fermi level~\cite{BottIndexUnitary2022toniolo}
\begin{equation}
    P=\frac{1}{2 \pi i} \oint_{\Gamma} dz \left(z-H\right)^{-1}.
\end{equation}
The contour integral representation can be understood by considering $P|\Psi_{occ}\rangle=|\Psi_{occ}\rangle\frac{1}{2 \pi i} \oint_{\Gamma} dz \left(z-E_{occ}\right)^{-1}=|\Psi_{occ}\rangle$ and $P|\Psi_{unocc}\rangle=|\Psi_{unocc}\rangle\frac{1}{2 \pi i} \oint_{\Gamma} dz \left(z-E_{unocc}\right)^{-1}=0$.

Utilizing the contour integral representation and Eq.~\eqref{eqs:commutor_equa}, it is easy to obtain that
\begin{equation}
\label{eqs:intP}
\begin{aligned}
       &\parallel[e^{2\pi i \frac{f(X,Y,Z,\dots)}{g(L)}},P]\parallel \\&\le \frac{1}{2 \pi}\parallel[e^{2\pi i \frac{f(X,Y,Z,\dots)}{g(L)}},H]\parallel \oint_{\Gamma} \parallel\left(H-z\right)^{-1}\parallel^2 |dz|,  
\end{aligned}
\end{equation}
with $\parallel\left(H-z\right)^{-1}\parallel^{2}= [\operatorname{dist}\left(z, \sigma(H)\right)]^{-2}$. $\operatorname{dist}\left(z, \sigma(H)\right)$ denotes the distance from a point $z$ to the region $\sigma(H)$. Taking the radius of $\Gamma$ to $\infty$, the loop-integral becomes
    \begin{equation}
    \label{eqs:loop-int}
    \begin{aligned}
           \oint_{\Gamma} \parallel\left(H-z\right)^{-1}\parallel^2 |dz| &=\oint_{\Gamma} [\operatorname{dist}\left(z, \sigma(H)\right)]^{-2} |dz| \\&=\int_{-\infty}^{\infty} \frac{1}{(\frac{\Delta E}{2})^2+(y)^2}dy\\&=\frac{2 \pi}{\Delta E}, 
    \end{aligned}
\end{equation}
where $y$ denotes $\operatorname{Im}z$.

Combining Theorem~\ref{thm: hamiNorm}, Eqs.~\eqref{eqs:intP} and ~\eqref{eqs:loop-int}, we have
$$
        \parallel[e^{2\pi i \frac{f(X,Y,Z,\dots)}{g(L)}},P]\parallel\le \mathcal{O}\left(\frac{R}{L} \frac{\parallel H\parallel}{\Delta E} \right).
$$
\end{proof}
With this Lemma proved, we obtain the following corollary (Theorem~\ref{thm:nozero} in the main text).
\begin{sectioncorollary}
\label{coro:nozero}
    If no zero-energy corner states exist in a system with chiral symmetry, finite coupling range, and energy gap, then the following conclusion holds true in the limit of large system size $L\to \infty$, with $\hat{M}=e^{2\pi i \frac{f(X,Y,Z,\dots)}{g(L)}}$.
\begin{equation}
\begin{aligned}
&-1\notin \sigma(\hat{M}q \hat{M}^{\dagger}q^{\dagger}),\\
&\operatorname{Bott}\left(\hat{M},q\right)=0,
\end{aligned}
\end{equation}
for arbitrary polynomial $f$ and $g$ with $\operatorname{deg}(f)=\operatorname{deg}(g)$. $q=U_AU_B^\dagger$.
\end{sectioncorollary}
\begin{proof}
    From Lemma~\ref{lem: normP} and projector $P=(\mathbbm{1}-H_{\text{Flat}})/2$, which in the eigenbasis of the chiral operator can be written as
    \begin{equation}
    P=\left( \begin{matrix}
	\frac{\mathbbm{1}}{2}&		-\frac{q}{2}\\
	-\frac{q^{\dagger}}{2}&		\frac{\mathbbm{1}}{2}\\
\end{matrix} \right),
\end{equation}
we have
    \begin{equation}
        \parallel[e^{2\pi i \frac{f(X,Y,Z,\dots)}{g(L)}},q]\parallel\le \mathcal{O}\left(\frac{R}{L} \frac{\parallel H\parallel}{\Delta E} \right).
    \end{equation}
    We note that in this equation the position operators are represented in the subspace of chiral operator.
    It follows that
    \begin{equation}
        -1\notin \sigma(\hat{M}q \hat{M}^{\dagger}q^{\dagger}).
    \end{equation}
    Also, when $L\to \infty$, we have $\parallel[e^{2\pi i \frac{f(X,Y,Z,\dots)\times t}{g(L)}},q]\parallel<2$ for all $t\in[0,1]$. From Corollary~\ref{coro: existence}, it follows that
    $$
    \operatorname{Bott}\left(e^{2\pi i \frac{f(X,Y,Z,\dots)}{g(L)}},q\right)=0.
    $$
\end{proof}
Thus, we provide the proof of the corollary~\ref{coro:nozero}, which is the Theorem~\ref{thm:nozero} in the main text.

\bigskip
\noindent \textbf{Detailed derivations of two equations in the main text}\\
In the following, we provide detailed derivation of Equation~\eqref{eqs:diagonal}. In an $n$-dimensional system under OBC, typically there are $\left(n-i\right)$-dimensional states ($i=0,1,\ldots,n$). The wavefunction of an $\left(n-i\right)$-dimensional state can be expressed as
\begin{equation}
    \left|\Psi_{n-i}\right\rangle=\sum_{{\v{x}}\in V_{n-i}}{a_{n-i,{\v{x}}}c_{\v{x}}^\dagger}\left|0\right\rangle+\sum_{{\v{x}}\notin V_{n-i}}a_{n-i,{\v{x}}}e^{-\beta|\Delta {\v{x}}|}c_{\v{x}}^\dagger\left|0\right\rangle,
\end{equation}
where $V_{n-i}$ is a region with $\mathcal{O}\left(L^{n-i}\right)$ hypervolume, and $|\Delta {\v{x}}|$ is the distance from ${\v{x}}$ to the position where the exponential decay begins. We can choose $V_{n-i}$ such that for ${\v{x}}\notin V_{n-i}$, $e^{-\beta|\Delta {\v{x}}|}<\mathcal{O}\left(L^n\right)$. Given that ${\langle\Psi}_{n-i}\left|\Psi_{n-i}\right\rangle=1$, we have $a_{n-i,{\v{x}}}=\mathcal{O}\left(\sqrt{\frac{1}{L^{n-i}}}\right)$. Thus, without loss of generality, we assume that $j>i$,
\begin{equation}
\begin{aligned}
        &{\langle\Psi}_{n-j}\left|\Psi_{n-i}\right\rangle \\&=\sum_{{\v{x}}\in V_{n-i}}{a_{n-j,{\v{x}}}^{\ast}a_{n-i,{\v{x}}}}\\&+\sum_{{\v{x}}\in V_{n-j},\ \ {\v{x}}\notin V_{n-i}}{a_{n-j,{\v{x}}}^{\ast}a_{n-i,{\v{x}}}e^{-\beta|\Delta {\v{x}}_i|}}\\&+\sum_{{\v{x}}\notin V_{n-j},\ \ {\v{x}}\notin V_{n-i}}{a_{n-j,{\v{x}}}^{\ast}a_{n-i,{\v{x}}}e^{-\beta|\Delta {\v{x}}_i|}}e^{-\beta|\Delta {\v{x}}_j|}\\ &=\mathcal{O}(\sqrt{\frac{1}{L^{j-i}}})\overset{L\to \infty}{\rightarrow}0.
\end{aligned}
\end{equation}
Similarly, we have
\begin{equation}
\begin{aligned}
&{\langle\Psi}_{n-j}\left|\hat{M}|\Psi_{n-i}\right\rangle\\&=\sum_{{\v{x}}\in V_{n-i}}{a_{n-j,{\v{x}}}^{\ast}a_{n-i,{\v{x}}}}e^{2\pi i\frac{f\left({\v{x}}\right)}{g\left(L\right)}}\\&+\sum_{{\v{x}}\in V_{n-j},\ \ {\v{x}}\notin V_{n-i}}{a_{n-j,{\v{x}}}^{\ast}a_{n-i,{\v{x}}}{e^{2\pi i\frac{f\left({\v{x}}\right)}{g\left(L\right)}}e}^{-\beta|\Delta{\v{x}}_i|}}\\&+\sum_{{\v{x}}\notin V_{n-j},\ \ {\v{x}}\notin V_{n-i}}{a_{n-j,{\v{x}}}^{\ast}a_{n-i,{\v{x}}}e^{2\pi i\frac{f\left({\v{x}}\right)}{g\left(L\right)}}e^{-\beta|\Delta{\v{x}}_i|}}e^{-\beta|\Delta{\v{x}}_i|}\\&=\mathcal{O}(\sqrt{\frac{1}{L^{j-i}}})\overset{L\to \infty}{\rightarrow}0.
\end{aligned}
\end{equation}

In the following, we show that in the limit $L\to \infty$ corner states will become eigenstates of the normalized position operator $\v{X}/L$ (Equation~\eqref{eqs:corner_as_eigen}).

The wavefunction of a corner state can be written as $\Psi_{\text{corner}}=\sum_{{\v{x}}}{a_{\v{x}} e^{-\beta\left|\v{{\v{x}}}-\v{{\v{x}}}_c\right|}c_{\v{x}}^\dagger\left|0\right\rangle}$, where $\v{x}_c$ is the position vector at which the corner state is localized, $a_{\v{x}}$ is the component at position $x$, and $\beta>0$ denotes the decay length. 
It follows that
\begin{equation}
\begin{aligned}
        &\frac{\v{X}}{L}\Psi_{\text{corner}} \\&=\sum_{\v{x}}{a_{\v{x}}{\frac{\v{x}}{L}e}^{-\beta\left|\v{x}-\v{x}_{c}\right|}c_{\v{x}}^\dagger\left|0\right\rangle} \\
        &=\frac{\v{x}_c}{L}\sum_{\v{x}}{a_{\v{x}} e^{-\beta\left|\v{x}-\v{x}_{c}\right|}c_{\v{x}}^\dagger\left|0\right\rangle}+\sum_{\v{x}}a_{\v{x}}\frac{\Delta \v{x}}{L}e^{-\beta|\Delta \v{x}|}c_{\v{x}}^\dagger\left|0\right\rangle \\
        &=\frac{\v{x}_c}{L}\sum_{\v{x}}{a_{\v{x}}e^{-\beta\left|\v{x}-\v{x}_c\right|}c_{\v{x}}^\dagger\left|0\right\rangle}+\mathcal{O}\left(1/L\right)+\mathcal{O}\left(e^{-\beta \mathcal{O}\left(L\right)}\right) \\
        & \overset{L\to \infty}{\rightarrow} \frac{\v{x}_c}{L}\sum_{\v{x}}{a_{\v{x}} e^{-\beta\left|\v{x}-\v{x}_{c}\right|}c_{\v{x}}^\dagger\left|0\right\rangle}=\frac{\v{x}_c}{L}\Psi_{\text{corner}}.
\end{aligned}
\end{equation}

\bigskip
\noindent \textbf{Shaped-dependent partially periodic boundary}\\ 
In this section, we describe the boundary conditions applied in our sum rule analysis. In an $n$-dimensional space, specifying the boundary conditions along $n$ linearly independent vectors fully determines those for the entire system. However, as noted in the main text, the choice of boundary conditions depends on the geometry of the system, which may seem to contradict the earlier statement. To clarify, the boundary conditions we discuss here focus on how the system's boundaries are ``glued" together.

For instance, applying periodic boundary conditions in an even-sided system can be interpreted as gluing together all pairs of parallel sides~\cite{davis2011periodictrajectoriesregularpentagon}. For a polygon with $2m+1$ sides, we can first duplicate the polygon and then identify the parallel sides~\cite{davis2011periodictrajectoriesregularpentagon,connor2021hyperbolicstaircasesperiodicpaths}.

By selectively gluing certain boundaries while leaving others intact, we effectively remove specific boundary conditions, enabling us to focus on the topological properties of both the bulk and the boundary states that remain localized at the retained boundaries.

\begin{algorithm}[H]
\caption{Generating Polynomials and the Associated Matrix $\mathcal{M}$ for an $m$-sided Shape}
\label{alg:polynomials}
   \begin{algorithmic}[1]
   \Statex \textbf{Input:} $m \in \mathbb{N}$ (number of polygon sides); $\v{x}_1^{(c)}, \dots, \v{x}_m^{(c)}$ (position vectors of the corners); $n$ (dimension).
   \Statex \textbf{Output:} Polynomials $f_{m}^{(1)}, \dots, f_{m}^{(m-1)}$ and corresponding polynomial $g(L)$; associated matrix $\mathcal{M}$ such that $\operatorname{det}(\mathcal{M}) \neq 0$.  
   
   \Procedure{Generate}{$f_{m}^{(i)}$,$g_m(L)$, and $\mathcal{M}$}
   
   \State Choose an $(m-1) \times (m-1)$ matrix $B$ with entries $B_{ij} = \pm 1$ such that
   \[
   \det\begin{pmatrix}
   A & B \\
   1 & D
   \end{pmatrix} \neq 0,
   \]
   where 
   \[
   A = [1, 1, \dots, 1]^T, \quad D = [1, 1, \dots, 1]^T.
   \]
   \State Define a homogeneous polynomial in the position operators as
   \begin{align*}
       p^{(l)}(X_1,\dots,X_{n}) = \sum_{\alpha \in \mathbb{N}, |\alpha|=l} a_{\alpha} X^{\alpha_1}_{1}X^{\alpha_2}_{2}\cdots X^{\alpha_n}_{n},
   \end{align*}
   where $a_{\alpha} \in \mathbb{R}, \quad l \in \mathbb{N}^{\ast}$
   \State Initialize $r \gets 1$ (solution counter) and $l \gets 1$ (degree of the polynomial).
   \While{$r \leq m-1$}
       \State Solve for the coefficients $a_{ij}$ using the equation
       \[
       \frac{\{ p^{(l)}(\v{x}_2^{(c)}), \dots, p^{(l)}(\v{x}_m^{(c)}) \}}{p^{(l)}(\v{x}_1^{(c)})} = \{ B_{r,1}, \dots, B_{r,m-1} \}.
       \]
       \If{the system of equations is redundant}
           \State Set selected coefficients $a_{ij}$ to remove redundancy.
       \EndIf
       \If{no solution is found for the current degree $l$}
           \State Increment $l \gets l + 1$ and repeat Step 6.
       \Else
           \State Set $f_m^{(r)} \gets p^{(l)}$ and 
                 \[
                    g_m^{(r)}(L) \gets \left| f_m^{(r)}(\v{x}_1^{(c)}) \right|.
                 \]
           \State Increment $r \gets r + 1$.
       \EndIf
   \EndWhile
   \State Set the matrix $\mathcal{M}$
   \[
   \mathcal{M} \gets \frac{1}{2} \begin{pmatrix}
   A & B \\
   1 & D
   \end{pmatrix}.
   \]
   \State \Return The set of $m-1$ pairs $\{ f_m^{(r)}, g_m^{(r)}(L) \}$ and the matrix $\mathcal{M}$.
   \EndProcedure
\end{algorithmic}
\end{algorithm}

\bigskip
\noindent \textbf{Procedure for generating polynomials and the associated matrix $\mathcal{M}$}  
In this section, we provide the algorithm for generating polynomials and their associated matrix $\mathcal{M}$, which satisfies $\operatorname{det}(\mathcal{M}) \neq 0$ for an $m$-cornered shape in $n$-dimensional space.

The matrix $\mathcal{M}$ is defined (as described in the main text) as
\begin{equation}
\mathcal{M}_{ij} = \begin{cases}
    \frac{\operatorname{sgn}(f_{m}^{(i)}(\boldsymbol{x}_j)/g_{m}^{(i)}(L))}{2}, & \text{for } 1 \le i \le m-1, \\
    \frac{1}{2}, & \text{for } i = m.
\end{cases}
\end{equation}
The algorithm to generate the polynomials and the corresponding matrix $\mathcal{M}$ is presented in Algorithm~\ref{alg:polynomials}.

\bigskip

\noindent{\bf{\large{RESOURCE AVAILABILITY}}}

\bigskip

\noindent{\large{\bf{Lead contact}}}

\noindent Further information and requests should be directed to and will be fulfilled by Fengcheng Wu and Meng Xiao.
\bigskip

\noindent{\large{\bf{Materials availability}}}

\noindent This study did not generate new materials.

\bigskip
\noindent{\large{\bf{Data and code availability}}}

\noindent \begin{itemize}
    \item All the theoretical data generated in this study will be deposited at Zenodo and publicly available as of the date of publication.
    \item The computer codes used to generate the results reported in this paper are available from the lead contact upon request.
\end{itemize}
 
\bigskip

\noindent \textbf{ACKNOWLEDGEMENTS}\\
J.-Z. L. and M.X. are supported by the National Key Research and Development Program of China (Grant No. 2022YFA1404900), the National Natural Science Foundation of China (Grant No. 12334015, Grant No. 12274332 and Grant No. 12321161645).
X.-J. L. and F.W. are supported by Key Research and Development Program of Hubei Province (Grant No. 2022BAA017). 

\bigskip

\noindent \textbf{AUTHOR CONTRIBUTIONS}\\
J.Z.L., X.J.L., F.W., and M.X. conceived the idea.  F.W., and M.X. supervised the project. J.Z.L. and X.J.L. did the theoretical analysis. J.Z.L., X.J.L., F.W., and M.X. wrote the manuscript.

\bigskip
\noindent \textbf{DECLARATION OF INTERESTS}\\
The authors declare no competing interests.

\widetext
\clearpage
\onecolumngrid   % if you prefer SI in two-column, delete this line
\beginsupplement

\begin{center}
{\large \textbf{Supplemental Information for\\
``Exact Universal Characterization of Chiral-Symmetric Higher-order Topological Phases''}}
\end{center}

The Supplemental Information provides further details on several aspects of this study. Specifically, it includes: a discussion of chiral symmetry with non-zero trace (Supplemental Information, Note 1); An analysis of systems exhibiting non-regular geometries (Supplemental Information, Note 2); a demonstration of the well-defined nature of the Bott index under open boundary conditions (Supplemental Information, Note 3); a compilation of polynomial series and their associated matrices for systems with diverse shapes (Supplemental Information, Note 4); a discussion of the real-space properties of the quadrupole moment and the multipole chiral number (Supplemental Information, Note 5); a discussion of the Bott index vector in systems possessing crystalline symmetry (Supplemental Information, Note 6); and additional lattice models with their corresponding calculations (Supplemental Information, Note 7).

\makeatletter
\renewcommand{\l@subsection}[2]{}
\renewcommand{\l@subsubsection}[2]{}
\makeatother

\tableofcontents

\section*{Supplemental Information, Note 1. Discussion of chiral symmetry with non-zero trace}
\addcontentsline{toc}{section}{Supplemental Information, Note 1. Discussion of chiral symmetry with non-zero trace}

In this section, we demonstrate that the topological classification of $n$-th order topological phases in $n$-dimensional systems possessing chiral symmetry with non-zero trace is trivial.

For the situation where the chiral operator has non-zero trace, the chiral operator is called general sublattice operator $S$ with non-zero trace~\cite{dai2024topologicalclassificationchiralsymmetry}, which can be represented as
\begin{eqnarray}
    S=\begin{pmatrix}\mathbbm{1}_M&0\\0&-\mathbbm{1}_N \end{pmatrix},
\end{eqnarray}
with $\mathbbm{1}_N$ being the $N\times N$ identity matrix. Without loss of generality, we assume $N\ge M$.
A Hamiltonian with such symmetry can be written as
\begin{equation}
H=\begin{pmatrix}0&h\\h^\dag&0\end{pmatrix},
\end{equation}
where $h$ is an $M \times N$ matrix. 
For this situation, the topological classification has been done by Ref.~\cite{dai2024topologicalclassificationchiralsymmetry}. 
First, since the $\operatorname{Rank}(h)\le N$, the dimension of zero space of $h$ is not less than $N-M$, implying that, in general, there exist $N-M$ flat bands at zero energy. Furthermore, the classifying spaces for this situation have been derived, assuming two energy gaps below and above the zero energy. The classifying space for the chiral symmetry with non-zero trace is the complex Stiefel manifold.
\begin{equation}
V_M(\mathbb{C}^N)=U(N)/U(N-M).
\end{equation}
The stable first homotopy group of $V_M(\mathbb{C}^N)$ is given as follows
\begin{equation}
    \pi_1(V_M(\mathbb{C}^N))= \begin{cases}
        \mathbb{Z} & \text{ if } N-M=0 \\
        0 & \text{ if } N-M\ge 1
    \end{cases}.
\end{equation}
Thus, for one-dimensional systems possessing chiral symmetry with non-zero trace, the topological classification is trivial. Similarly, we can deduce that the topological classification for the $n$-th order topological phases in $n$-dimensional systems with chiral symmetry and non-zero trace is trivial.

\section*{Supplemental Information, Note 2. Discussion for systems with non-regular shapes}
\addcontentsline{toc}{section}{Supplemental Information, Note 2. Discussion for systems with non-regular shapes}

For systems featuring a non-regular shape with $m$ corners, we employ a mapping, $\mathcal{F}$, that converts them into their regular counterparts while preserving the corners, by utilizing the generalized Barycentric coordinates~\cite{hormann2006mean,floater2015generalized}. We then define polynomials $f$ and $g$ as $f^{\prime}\circ \mathcal{F}$ and $g^{\prime}\circ \mathcal{F}$, respectively. Here, $\circ$ denotes function composition. $f^{\prime}$ and $g^{\prime}$ are the corresponding polynomials for an $m$-cornered regular shape. $f^{\prime}$ and $g^{\prime}$ are derived and provided in the previous section.

\section*{Supplemental Information, Note 3. The well-defined property of Bott index under open boundary conditions}
\addcontentsline{toc}{section}{Supplemental Information, Note 3. The well-defined property of Bott index under open boundary conditions}

Under open boundary conditions, zero-energy corner states may appear in a system, meaning that $\det{h}=0$ and $\mathrm{rank}\left(\ker (h)\right)\neq 0$. $\ker (h)$ denotes the kernel of $h$, which is the vector space spanned by vectors satisfying $h\cdot\v{v}=\v{0}$, defined as $\ker (h)=\{\v{v}|h\cdot\v{v}=\v{0}\}$. Thus, the operator $q$ is not unique due to arbitrary unitary transformations applied to vectors of $U_A$ and $U_B$ that span the cokernel and kernel of $h$~\cite{HandbookLinearAlgebra2013hogben}. The cokernel of $h$ equals $\ker (h^{\dagger})$. Thus, these transformations can be represented by arbitrary unitary transformations $W_A$ on $U_{A,\text{corner}}$ and $W_B$ on $U_{B,\text{corner}}$. 

With the definition of $\hat{M}$ in the main text, we have
\begin{equation}
\begin{aligned}
W_{\alpha}^{\dagger}U_{\alpha,\text{corner}}^{\dagger}\hat{M}U_{\alpha,\text{corner}}W_{\alpha} &=-W_{\alpha}^{\dagger}U_{\alpha,\text{corner}}^{\dagger}U_{\alpha,\text{corner}}W_{\alpha}\\&=-\mathbbm{1}.
\end{aligned}
\end{equation}
This agrees with the result obtained in the main text when there is no transformation, regardless of arbitrary transformation. 

\section*{Supplemental Information, Note 4. The series of polynomials and associated matrices $\mathcal{M}$ for systems with ${\{4,5,6,8,12\}}$-sided regular polygon shapes}
\addcontentsline{toc}{section}{Supplemental Information, Note 4. The series of polynomials and associated matrices $\mathcal{M}$ for systems with ${\{4,5,6,8,12\}}$-sided regular polygon shapes}

In this section, we follow this procedure in Methods section of main text and derive the series of polynomials for systems with ${\{4,5,6,8,12\}}$-sided regular polygon shapes. The order of corners in the configuration vector is selected in a counterclockwise order, starting from the $x$-axis. $L$ represents the side length of the system. The coordinate origin is always placed at the geometric center of each system with regular polygon shapes. 

For systems with a square shape, we choose the $x$-axis pointing toward the midpoint between two corners and the series of polynomials $f/g$ as follows
\begin{small}
\begin{equation}
        \left\{2XY/L^2, X/L, Y/L \right\},
\end{equation}
\end{small}
with 
\begin{small}
\begin{equation}
\mathcal{M}=\frac{1}{2}\begin{pmatrix}
    1 & -1 & 1 &-1 \\
    1 & -1 & -1 &1 \\
    1 & 1 & -1 &-1 \\
    1 & 1 & 1 &1
\end{pmatrix}.
\end{equation}
\end{small}

For systems with $5$-sided regular polygon shape, we choose the $x$-axis pointing to a corner and the series of polynomials $f/g$ as follows
\begin{small}
\begin{equation}
\begin{aligned}
   \{&X^4-4 \sqrt{1+\frac{2}{\sqrt{5}}} X^3 Y-2 X^2 Y^2+\frac{4}{5} \sqrt{25+2
   \sqrt{5}} X Y^3+\frac{Y^4}{5},\\&X^4+4 \sqrt{1-\frac{2}{\sqrt{5}}} X^3 Y-2 X^2
   Y^2+\frac{4}{5} \sqrt{25-2 \sqrt{5}} X Y^3+\frac{Y^4}{5},\\&X^4-\frac{2}{5}
   \left(5+6 \sqrt{5}\right) X^2 Y^2+\frac{1}{5} \left(1+4 \sqrt{5}\right)
   Y^4,\\&X^4+\left(\frac{12}{\sqrt{5}}-2\right) X^2 Y^2+\frac{1}{5} \left(1-4
   \sqrt{5}\right) Y^4\}/2(\frac{L}{2\sin \pi/5})^4
\end{aligned}
\end{equation}
\end{small}
with
\begin{small}
\begin{equation}
\mathcal{M}=\frac{1}{2}\begin{pmatrix}
 1 & 1 & 1 & -1 & -1 \\
 1 & 1 & -1 & 1 & -1 \\
 1 & 1 & -1 & -1 & 1 \\
 1 & -1 & 1 & 1 & -1 \\
 1 & 1 & 1 & 1 & 1 \\
\end{pmatrix}.
\end{equation}
\end{small}

For systems with a regular hexagon shape, we choose the $x$-axis pointing to a corner and the series of polynomials $f/g$ as follows
\begin{equation}
\begin{aligned}
     \{&(X^3-\frac{XY^2}{3}+\frac{8Y^3}{3\sqrt{3}})/2L^3, (X^2-\frac{4XY}{\sqrt{3}}-\frac{Y^2}{3})/2L^2,\\ &(X^2+\frac{4XY}{\sqrt{3}}-\frac{Y^2}{3})/2L^2,(X^3-3XY^2)/2L^3,\\ &(X^3+\frac{7XY^2}{3})/2L^3\},
\end{aligned}
\end{equation}
with
\begin{small}
\begin{equation}
\mathcal{M}=\frac{1}{2}\begin{pmatrix}
 1 & 1 & 1 & -1 & -1 & -1 \\
 1 & -1 & 1 & 1 & -1 & 1 \\
 1 & 1 & -1 & 1 & 1 & -1 \\
 1 & -1 & 1 & -1 & 1 & -1 \\
 1 & 1 & -1 & -1 & -1 & 1 \\
    1 & 1 & 1 &1 & 1 &1
\end{pmatrix}.
\end{equation}
\end{small}
For systems with $8$-sided regular polygon shape, we choose the $x$-axis pointing to a corner and the series of polynomials $f$ as follows
\begin{small}
\begin{equation}
\begin{aligned}
     \{&X^3-X^2 Y+\left(2 \sqrt{2}-1\right) X Y^2+Y^3,\\ &-X^3-\left(1+2 \sqrt{2}\right)
   X^2 Y+X Y^2+Y^3,\\ &-X^3-X^2 Y+\left(1-2 \sqrt{2}\right) X Y^2+Y^3,\\ &X^3-X^2
   Y-\left(1+2 \sqrt{2}\right) X Y^2+Y^3,\\ &-X^4-4 X^3 Y+Y^4,\\ &X^4-6 X^2 Y^2+Y^4,\\ &-X^4+4
   X^3 Y+Y^4\},
\end{aligned}
\end{equation}
\end{small}
and $g=2(\frac{L}{2\sin \pi/8})^{\operatorname{deg}(f)}$,
with
\begin{small}
\begin{equation}
\mathcal{M}=\frac{1}{2}\begin{pmatrix}
1 & 1 & 1 & -1 & -1 & -1 & -1 & 1 \\
 -1 & -1 & 1 & -1 & 1 & 1 & -1 & 1 \\
 -1 & -1 & 1 & 1 & 1 & 1 & -1 & -1 \\
 1 & -1 & 1 & 1 & -1 & 1 & -1 & -1 \\
 -1 & -1 & 1 & 1 & -1 & -1 & 1 & 1 \\
 1 & -1 & 1 & -1 & 1 & -1 & 1 & -1 \\
 -1 & 1 & 1 & -1 & -1 & 1 & 1 & -1 \\
 1 & 1 & 1 & 1 & 1 & 1 & 1 & 1 \\
\end{pmatrix}.
\end{equation}
\end{small}
For systems with $12$-sided regular polygon shape, we choose the $x$-axis toward the midpoint between two corners and the series of polynomials $f/g$ as:
\begin{small}
\begin{eqnarray}
\{&&(-6 x^5y + 20x^3 y^3 - 6 xy^5)/2L^6,\nonumber\\
&&-(2/9) (\sqrt{3} x^6 - 21 x^5 y - 15 \sqrt{3} x^4 y^2 +  6 x^3 y^3 + 15 \sqrt{3} x^2 y^4 - 21 x y^5 - \sqrt{3} y^6)/2L^6,\nonumber\\
&&1/3 (3 x^6 -  (8 + 4\sqrt{3}) x^5 y - (15 + 16\sqrt{3}) x^4 y^2 +16 x^3 y^3 + (-15 + 16 \sqrt{3}) x^2 y^4+  (-8 + 4\sqrt{3}) x y^5 + 3y^6)/2L^6,\nonumber\\
&&-(2/9) (2 \sqrt{3} x^6 - 3 x^5 y + 18 \sqrt{3} x^4 y^2 + 42 x^3 y^3 - 18 \sqrt{3} x^2 y^4 - 3 x y^5 - 2 \sqrt{3} y^6)/2L^6,\nonumber\\
&&((9 + \sqrt{3}) x^6 + (24 + 12\sqrt{3}) x^5 y -(45 + 63 \sqrt{3}) x^4 y^2 - 48 x^3 y^3 -  (45 - 63 \sqrt{3}) x^2 y^4 + (24 -12 \sqrt{3}) x y^5 + (9 - \sqrt{3}) y^6)/18L^6,\nonumber\\
&&-\sqrt{2} y ((-1 + 6 \sqrt{3}) x^4 +  2 (7 - 4 \sqrt{3}) x^2 y^2 + (-1 + 2 \sqrt{3}) y^4)/6L^5,\nonumber\\
&&\sqrt{2} ((1 + \sqrt{3}) x^5 + 3 (-2 + \sqrt{3}) x^4 y -4 (-1 + \sqrt{3}) x^3 y^2 + 2 (9 - 2 \sqrt{3}) x^2 y^3 + (-5 + 3 \sqrt{3}) x y^4 + \sqrt{3} y^5)/6L^5,\nonumber\\
&&-\sqrt{2} (x^5 - 6 \sqrt{3} x^4 y - 14 x^3 y^2 + 8 \sqrt{3} x^2 y^3 + x y^4 - 2 \sqrt{3} y^5)/6L^5,\nonumber\\
&&1/3 \sqrt{2} ((-1 + \sqrt{3}) x^5 + (4 + 3 \sqrt{3}) x^4 y - 4 (1 + \sqrt{3}) x^3 y^2 +  2 (5 - 2 \sqrt{3}) x^2 y^3 + (5 + 3 \sqrt{3}) x y^4 + (-2 + \sqrt{3}) y^5)/2L^5,\nonumber\\
&&\sqrt{2} ((-1 + \sqrt{3}) x^5 - (4 + 3 \sqrt{3}) x^4 y - 4 (1 + \sqrt{3}) x^3 y^2 + 2 (-5 + 2 \sqrt{3}) x^2 y^3 + (5 + 3 \sqrt{3}) xy^4 - (-2 +  \sqrt{3}) y^5)/6L^5,\nonumber\\
&&-(1/3) \sqrt{2} (x - \sqrt{3} y) (x^4 + 7 \sqrt{3} x^3 y + 7 x^2 y^2 - \sqrt{3} x y^3 - 2 y^4)/2L^5\},
\end{eqnarray}    
\end{small}
with
\begin{small}
\begin{equation}
\mathcal{M}=\frac{1}{2}\left(
\begin{array}{cccccccccccc}
 -1 & 1 & -1 & 1 & -1 & 1 & -1 & 1 & -1 & 1 & -1 & 1 \\
 1 & 1 & 1 & -1 & -1 & -1 & 1 & 1 & 1 & -1 & -1 & -1 \\
 -1 & -1 & 1 & 1 & -1 & 1 & -1 & -1 & 1 & 1 & -1 & 1 \\
 -1 & -1 & 1 & 1 & 1 & -1 & -1 & -1 & 1 & 1 & 1 & -1 \\
 1 & -1 & 1 & 1 & -1 & -1 & 1 & -1 & 1 & 1 & -1 & -1 \\
 -1 & -1 & -1 & -1 & -1 & -1 & 1 & 1 & 1 & 1 & 1 & 1 \\
 1 & 1 & 1 & 1 & 1 & -1 & -1 & -1 & -1 & -1 & -1 & 1 \\
 1 & 1 & 1 & 1 & -1 & 1 & -1 & -1 & -1 & -1 & 1 & -1 \\
 1 & 1 & 1 & -1 & 1 & 1 & -1 & -1 & -1 & 1 & -1 & -1 \\
 -1 & -1 & 1 & -1 & -1 & -1 & 1 & 1 & -1 & 1 & 1 & 1 \\
 -1 & 1 & -1 & -1 & -1 & -1 & 1 & -1 & 1 & 1 & 1 & 1 \\
 1 & 1 & 1 & 1 & 1 & 1 & 1 & 1 & 1 & 1 & 1 & 1 \\
\end{array}
\right).
\end{equation}
\end{small}

\section*{Supplemental Information, Note 5. Quadrupole moment $q_{xy}$ and multipole chiral number $N_{xy}$ in real space}
\addcontentsline{toc}{section}{Supplemental Information, Note 5. Quadrupole moment $q_{xy}$ and multipole chiral number $N_{xy}$ in real space}

\subsection{Definitions}
In this section, we provide the definitions of $q_{xy}$ and $N_{xy}$ used in our calculations of Fig.~2(a) of the main text.

The quadrupole moment defined in the real space is given by~\cite{TopologicalPhaseTransitions2020li,ManybodyElectricMultipole2019wheeler,ManybodyOrderParameters2019kanga}:
\begin{equation}
q_{x y}=\frac{1}{2 \pi} \operatorname{Im} \log \left[\operatorname{det}\left(V^{\dagger} \hat{Q} V\right) \sqrt{\operatorname{det}\left(\hat{Q}^{\dagger}\right)}\right] \quad \operatorname{mod} 1,
\end{equation}
where $V$ is a matrix composed of all the occupied eigenstates of $H$ under periodic boundary conditions. $\hat{Q}=e^{2\pi i \frac{XY}{L^2}}$.

Multipole chiral number, defined in 2D systems, is given by~\cite{ChiralSymmetricHigherOrderTopological2022benalcazar}:
\begin{equation}
    N_{xy}=\frac{1}{2\pi i}\operatorname{Tr}\operatorname{log}\left(V_A^{\dagger} \hat{Q}^{c} V_AV_B^{\dagger} \hat{Q}^{c \dagger} V_B\right),
\end{equation}
where $V_A$ and $V_B$ is determined by the singular value decomposition of $h$, $h=V_A\Sigma V_B^{\dagger}$, under the periodic boundary condition. $\hat{Q}^{c}=e^{2\pi i \frac{XY}{L^2}}$, with position operators represented in the subspace of chiral operator.

\subsection{The relationship between $q_{xy}$ and $N_{xy}$}
In this section, we prove the correspondence between $q_{xy}$ and $N_{xy}$ in chiral-symmetric systems,

\begin{equation}
\label{eqs:s_q_chiralrelationship}
    q_{xy}=\frac{N_{xy}}{2} \quad \operatorname{mod} 1.
\end{equation}

Since the following derivations do not contain the information of multipole operator, the conclusion can be generalized to higher multipole situations. 

By taking the eigenbasis of chiral operator, we rewrite $V$ in $q_{x y}$ as:
\begin{equation}
    V=\frac{1}{\sqrt{2}}\begin{pmatrix}
        V_A \\
        V_B
    \end{pmatrix},
\end{equation}
and rewrite $\hat{Q}$ in $q_{x y}$ as:
\begin{equation}
    \hat{Q}=\begin{pmatrix}
        \hat{Q}^{c} & 0 \\
        0 & \hat{Q}^{c}
    \end{pmatrix}.
\end{equation}

Therefore, we have:
\begin{equation}
\begin{aligned}
    q_{x y} &=\frac{1}{2 \pi} \operatorname{Im} \log \left[\operatorname{det}\left(V^{\dagger} \hat{Q} V\right) \sqrt{\operatorname{det}\left(\hat{Q}^{\dagger}\right)}\right]  \quad \operatorname{mod} 1 \\
    &=\frac{1}{2 \pi} \operatorname{Im} \log \left[\operatorname{det}\left(V_A^{\dagger} \hat{Q}^{c} V_A+V_B^{\dagger} \hat{Q}^{c} V_B\right) \sqrt{\operatorname{det}\left(\hat{Q}^{\dagger}\right)}\right] \quad \operatorname{mod} 1\\
    &=\frac{1}{2 \pi} \operatorname{Im} \log \left[\operatorname{det}\left(V_A^{\dagger} \hat{Q}^{c} V_A+\mathbbm{1}\right)\det \left(V_B^{\dagger} \hat{Q}^{c} V_B\right) \sqrt{\operatorname{det}\left(\hat{Q}^{\dagger}\right)}\right] \quad \operatorname{mod} 1\\
    &=\frac{1}{2 \pi} \operatorname{Im} \log \left[\operatorname{det}\left(V_A^{\dagger} \hat{Q}^{c} V_AV_B^{\dagger} \hat{Q}^{c\dagger} V_B+\mathbbm{1}\right)\right]  \quad \operatorname{mod} 1,
\end{aligned}
\end{equation}
where in the last step we use $\operatorname{det}\left(\hat{Q}^{c}\right) \sqrt{\operatorname{det}\left(\hat{Q}^{\dagger}\right)}=1$. 

Denoting the eigenvalues of $V_A^{\dagger} \hat{Q}^{c} V_AV_B^{\dagger} \hat{Q}^{c\dagger} V_B$ as $e^{i \lambda_i}$, we have 
\begin{equation}
    N_{xy}=\frac{1}{2\pi i}\operatorname{Tr}\operatorname{log}\left(V_A^{\dagger} \hat{Q}^{c} V_AV_B^{\dagger} \hat{Q}^{c\dagger} V_B\right)=\frac{\sum_{i}\lambda_i}{2 \pi}.
\end{equation}
We express the $q_{x y}$ as follows:
\begin{equation}
\begin{aligned}
        q_{x y} &=\frac{1}{2 \pi} \operatorname{Im} \log \left[\operatorname{det}\left(V_A^{\dagger} \hat{Q}^{c} V_AV_B^{\dagger} \hat{Q}^{c\dagger} V_B+\mathbbm{1}\right)\right]  \quad \operatorname{mod} 1\\
        &=\frac{1}{2 \pi} \operatorname{Im} \log \left[\prod_{i}\left(e^{i \lambda_i}+1\right)\right]  \quad \operatorname{mod} 1\\
        &=\frac{1}{2 \pi} \operatorname{Im}\sum_{i} \log \left(e^{i \lambda_i}+1\right)  \quad \operatorname{mod} 1\\
        &=\frac{1}{2 \pi} \operatorname{Im}\sum_{i} i \frac{\lambda_i}{2}  \quad \operatorname{mod} 1\\
        &=\frac{\sum_{i}\lambda_i}{4 \pi}  \quad \operatorname{mod} 1 \\
        &= \frac{N_{xy}}{2} \quad \operatorname{mod} 1.
\end{aligned}
\end{equation}

\subsection{Discussions}
Generally, both \( N_{xy} \) and \( q_{xy} \) are defined for systems with rectangular geometries. The structure of \( N_{xy} \) indicates a conceptual similarity to the Bott index \( \nu \) introduced in the main text, stemming from their shared foundation in K-theory, as both pertain to the AIII class. However, our approach diverges significantly from previous frameworks in both scope and impact. It critically examines the correspondences they establish, as highlighted in Table~\ref{tab:Compare}.

\begin{table}[]
    \centering
    \begin{tabular}{|p{4cm}|p{4cm}|p{7cm}|}
        \hline
        \multicolumn{1}{|c|}{\textbf{Scheme}} & 
        \multicolumn{1}{c|}{\textbf{Key Limitation}} & 
        \multicolumn{1}{c|}{\textbf{Evidence}} \\
        \hline
        Nested Wilson Loop \newline (Multipole Moment) 
        & Lacks full reliability and completeness. 
        & Recent works~\cite{PhysRevLett.132.213801,QuadrupoleInsulatorCorner2023tao,PhysRevResearch.2.043012,ChiralSymmetricHigherOrderTopological2022benalcazar}. \\
        \hline
        Polarization contributions \newline
        (e.g., $Q^{\mathrm{corner}}-p_{x}^{\mathrm{edge}}-p_{y}^{\mathrm{edge}}=-q_{xy}$) 
        & Lacks full reliability.
        & Type-II quadrupole topological insulators~\cite{TypeIIQuadrupoleTopological2020yang}. \\
        \hline
        Real-space Multipole Moment 
        & Lacks full reliability and completeness. 
        & Recent works~\cite{PhysRevLett.132.213801,QuadrupoleInsulatorCorner2023tao,PhysRevResearch.2.043012,ChiralSymmetricHigherOrderTopological2022benalcazar,DifficultiesOperatorbasedFormulation2019ono}. \\
        \hline
        Multipole Chiral Number 
        & Lacks full reliability and completeness. 
        & The first lattice model of main text (Eq. (30)) and recent works~\cite{PhysRevLett.132.213801,QuadrupoleInsulatorCorner2023tao} (utilizing Equation~\eqref{eqs:s_q_chiralrelationship}). \\
        \hline
    \end{tabular}
    \caption{Limitations of previous schemes and the corresponding supporting evidence.}
    \label{tab:Compare}
\end{table}

As shown by the definition of $N_{xy}$, it can be rewritten in the form of the Bott index,
\begin{equation}
N_{xy}=\operatorname{Bott}\left(\hat{Q}^{c},V_AV_B^{\dagger}\right).
\end{equation}
It should be noted that under periodic boundary conditions, the Bott index form of $N_{xy}$ may be ill-defined according to the Definition~M1 in the Method Section, with $\parallel [\hat{Q}^{c},V_AV_B^{\dagger}]\parallel$ equal to $2$. This is because the hopping, required by periodic boundary condition, is long-ranged, which renders Theorem~M2 in Methods inapplicable to Hamiltonians under such conditions. 
Let us illustrate this by considering 
$$e^{2\pi i \frac{x_my_m-x_ny_n}{L^2}},$$ 
which appears in the proof of Theorem~M2, Eq.~(40) when $f/g=XY/L^2$.
Choosing $x_m=1$, $x_n=L-R$, and $y_m=y_n=L/2$, 
we have 
$$e^{2\pi i \frac{x_my_m-x_ny_n}{L^2}}=-1,$$ 
when $R\ll L$. 
Thus, Eq.~(40) in the proof of Theorem~M2 is not satisfied. There is no upper bound decaying as $L$ increases for $\parallel [\hat{Q}^{c},V_AV_B^{\dagger}]\parallel$, leaving the possibility that $\parallel [\hat{Q}^{c},V_AV_B^{\dagger}]\parallel=2$.

\section*{Supplemental Information, Note 6. Bott index vector in systems with crystalline symmetries}
\addcontentsline{toc}{section}{Supplemental Information, Note 6. Bott index vector in systems with crystalline symmetries}

In this section, we discuss the constraints of certain crystalline symmetries on the Bott index vector.

First, we show that the co-existence of two ZECSs with the opposite chirality at the same corner is not topologically stable. Assuming that there are two ZECSs $|+\rangle$ and $|-\rangle$ with chirality $+$ and $-$ at one corner of a system. It is easy to verify that the perturbing term $\delta H=|+\rangle\langle-|+|-\rangle\langle+|$ gaps the two ZECSs out while keeping the chiral symmetry~\cite{SecondorderTopologicalInsulators2018geier}. Thus, such co-existence is not topologically stable.

\subsection*{The four-fold rotation symmetry}
Now, we argue that the $C_4$ rotation operator must anticommute with the chiral operator for a system to exhibit nontrivial HOTPs. If we assume that the four-fold rotation operator commutes with the chiral operator, we have four ZECSs ${\left(R_4\right)^i\Psi}_{corner}$ [$i=0,1,2,3$] located at each of the four corners of the square, where $\Psi_{corner}$ is a wavefunction of a ZECS located at one corner. All four of these ZECSs have the same chirality. Since $\sum_{i=1}^{4}\chi_i=0$, there must be a ZECS with the opposite chirality located at each corner of the square. However, as mentioned above, the coexistence of two ZECSs with opposite chiralities at the same corner is not topologically stable. Thus, the four-fold rotation operator must anticommute with the chiral operator if the system possesses non-trivial higher-order topological phases.

In the following, we prove that $\nu_x=\nu_y=0$ when an $C_4$ rotational symmetry is present.
First, we have
\begin{eqnarray}
    R_4 H\left(x,y\right)R_4^\dagger=H\left(-y,x\right), R_4=\begin{pmatrix} 0& R_4^A \\ R_4^B & 0 \end{pmatrix},
\end{eqnarray}
where the rotation operator is written in the eigenbasis of the chiral operator. It follows that
\begin{eqnarray}
R_4^BU_A&=\tilde{U}_B, \\
R_4^AU_B&=\tilde{U}_A,
\end{eqnarray}
where $\tilde{\cdot}$ denotes matrices and Bott indices obtained in the system rotated by $90$ degrees. Thus, we have
\begin{equation}
\begin{aligned}
\tilde{\nu}_x &=\operatorname{Bott}(e^{2\pi i \frac{-Y}{L}},q)={-\nu}_y \\
\tilde{\nu}_x &=\frac{1}{2\pi i}\mathrm{Tr\ Log}\left(U_B^\dagger\left(R_4^A\right)^\dagger e^{2\pi i\frac{X}{L}}R_4^AU_BU_A^\dagger\left(R_4^B\right)^\dagger e^{-2\pi i\frac{X}{L}}R_4^BU_A\right) \\
&=\frac{1}{2\pi i}\mathrm{Tr\ Log}\left(U_B^\dagger e^{2\pi i\frac{X}{L}}U_BU_A^\dagger e^{-2\pi i\frac{X}{L}}U_A\right) ={-\nu}_x.    
\end{aligned}
\end{equation}
Similarly, for $\nu_y$ we have
\begin{equation}
\begin{aligned}
\tilde{\nu}_y&=\operatorname{Bott}(e^{2\pi i \frac{X}{L}},q)=\nu_x \\
\tilde{\nu}_y&=\frac{1}{2\pi i}\mathrm{Tr\ Log}\left(U_B^\dagger\left(R_4^A\right)^\dagger e^{2\pi i\frac{Y}{L}}R_4^AU_BU_A^\dagger\left(R_4^B\right)^\dagger e^{-2\pi i\frac{Y}{L}}R_4^BU_A\right) \\&=\frac{1}{2\pi i}\mathrm{Tr\ Log}\left(U_B^\dagger e^{2\pi i\frac{Y}{L}}U_BU_A^\dagger e^{-2\pi i\frac{Y}{L}}U_A\right)={-\nu}_y,
\end{aligned}
\end{equation}
and for $\nu_{2xy}$ we have
\begin{equation}
\begin{aligned}
\tilde{\nu}_{2xy}&=\operatorname{Bott}(e^{2\pi i \frac{-2XY}{L^2}},q)={-\nu}_{2xy}\\
\tilde{\nu}_{2xy}&=\frac{1}{2\pi i}\mathrm{Tr\ Log}\left(U_B^\dagger\left(R_4^A\right)^\dagger e^{2\pi i\frac{2XY}{L^2}}R_4^AU_BU_A^\dagger\left(R_4^B\right)^\dagger e^{-2\pi i\frac{2XY}{L^2}}R_4^BU_A\right)\\& =\frac{1}{2\pi i}\mathrm{Tr\ Log}\left(U_B^\dagger e^{2\pi i\frac{2XY}{L^2}}U_BU_A^\dagger e^{-2\pi i\frac{2XY}{L^2}}U_A\right)={-\nu}_{2xy}.
\end{aligned}
\end{equation}
From the above results, we have $\nu_x=\nu_y=-\nu_x$. Thus, $\nu_x$ and $\nu_y$ must be equal to $0$. This result corresponds to the configuration vector of corner states in an $C_4$ rotation-symmetric system. With the equation
\begin{equation}
    \v{\chi}_4=\mathcal{M}^{-1}\cdot\left(\nu_{2xy},0,0,0\right),
\end{equation}
we obtain that the configuration vector equals $\v{\chi}_4=\left(N_1,-N_1,N_1,-N_1\right)^T$, which corresponds to the $C_4$ symmetry.

\section*{Supplemental Information, Note 7. More lattice models}
\addcontentsline{toc}{section}{Supplemental Information, Note 7. More lattice models}

In this section, we introduce two systems, one with a pentagon shape and the other with a cubic shape, to demonstrate our framework for arbitrary shapes and dimensions.
First, we consider a system shaped as a regular pentagon, with the corresponding Bloch Hamiltonian $H_{penta}$, expressed as follows:
\begin{equation}
    h_{penta}(\v{k})=\left( \begin{matrix}
	-t_x-t_x^{\prime} e^{-i k_x}&		t_y+t_y^{\prime} e^{i k_y}\\
	t_y+t_y^{\prime} e^{-i k_y}&		t_x+t_x^{\prime} e^{i k_x}\\
\end{matrix} \right)+2w_d\left( \begin{matrix}
	 e^{-i k_x}\cos k_y&		- e^{i k_y}\cos k_x\\
	- e^{-i k_y}\cos k_x&		- e^{i k_x}\cos k_y\\
\end{matrix} \right).   
\end{equation}
This Bloch Hamiltonian is originally introduced in Ref.~\cite{ChiralSymmetricHigherOrderTopological2022benalcazar} and defined on rectangle lattices. However, by cutting the boundary of this Hamiltonian, we obtain a system with a regular pentagon shape under open boundary conditions. We characterize the pattern of ZECSs in this system by calculating the Bott index vector in real space, using the following $f/g$ obtained in the previous section.
\begin{equation}
    \begin{aligned}
   \{&X^4-4 \sqrt{1+\frac{2}{\sqrt{5}}} X^3 Y-2 X^2 Y^2+\frac{4}{5} \sqrt{25+2
   \sqrt{5}} X Y^3+\frac{Y^4}{5},\\&X^4+4 \sqrt{1-\frac{2}{\sqrt{5}}} X^3 Y-2 X^2
   Y^2+\frac{4}{5} \sqrt{25-2 \sqrt{5}} X Y^3+\frac{Y^4}{5},\\&X^4-\frac{2}{5}
   \left(5+6 \sqrt{5}\right) X^2 Y^2+\frac{1}{5} \left(1+4 \sqrt{5}\right)
   Y^4,\\&X^4+\left(\frac{12}{\sqrt{5}}-2\right) X^2 Y^2+\frac{1}{5} \left(1-4
   \sqrt{5}\right) Y^4\}/2(\frac{L}{2\sin \pi/5})^4.
\end{aligned}
\end{equation}
The calculations show that the Bott index vector equals $(0,2,0,0)$, and the configuration vector $\v{\chi}_5$ equals $(0,1,-1,1,-1)^{T}$ for $t_x=t_y=1$, $t_x^{\prime}=t_y^{\prime}=3$, and $w_d=1$, as illustrated in the Fig.~\ref{fig:Sup-Fig1_penta_3d}. They follows the exact relationship
\begin{equation}
     \boldsymbol{\chi}_{5}=\mathcal{M}^{-1}\cdot \left(\nu^{(1)}_5,\dots,\nu^{(4)}_5,0\right)^{\mathrm{T}},
\end{equation}
with $\mathcal{M}$ provided in the previous section.

\begin{figure}
\centering
	\includegraphics[width=1.0\columnwidth]{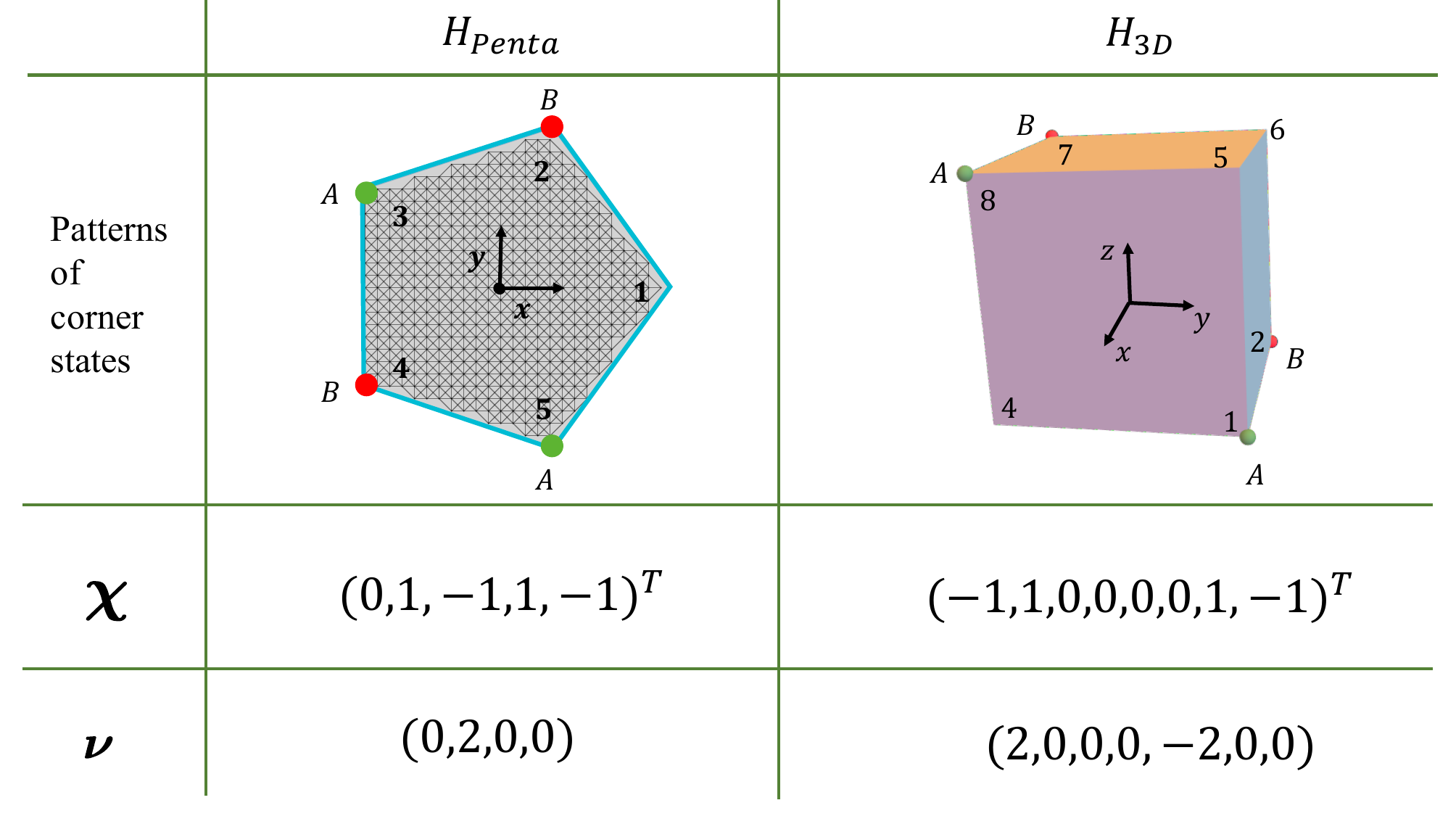}
	\caption {The two systems are shaped as a regular pentagon and a cube, respectively. The Bott index vector $\v{\nu}$ and configuration vectors $\v{\chi}$ of these systems adhere to the analytical relationship outlined in the main text. For $H_{penta}$, the system size $L$ is $30$, and for $H_{3D}$, it is $15$.}
	\label{fig:Sup-Fig1_penta_3d}
\end{figure}

Second, we examine a model with a cubic shape.
\begin{equation}
\begin{aligned}
        h_{3D}(\v{k})=&(m_x-i\gamma_x \sin k_x) \mathbbm{1}+ (t_z+\gamma_z \cos k_z)\tau_x\otimes\tau_0-\gamma_z\sin k_z \tau_y\otimes\tau_0\\&-(t_y-\gamma_y\cos k_y)\tau_z\otimes\tau_x+\gamma_y \sin k_y \tau_x\otimes \tau_y + (t_x+\gamma_x\cos k_x) \tau_z \otimes \tau_z.
\end{aligned}
\end{equation}
All possible patterns of HOTPs in this system can be captured by the Bott index vector defined by the following $f/g$,
\begin{equation}
    \{\frac{4XYZ}{L^3},\frac{2XY}{L^2},\frac{2XZ}{L^2},\frac{2YZ}{L^2},\frac{X}{L},\frac{Y}{L},\frac{Z}{L}\}
\end{equation}
with 
\begin{equation}
    \mathcal{M}=\frac{1}{2}\begin{pmatrix}
     -1 & 1 & -1 & 1 & 1 & -1 & 1 & -1 \\
 1 & -1 & 1 & -1 & 1 & -1 & 1 & -1 \\
 -1 & 1 & 1 & -1 & 1 & -1 & -1 & 1 \\
 -1 & -1 & 1 & 1 & 1 & 1 & -1 & -1 \\
 1 & -1 & -1 & 1 & 1 & -1 & -1 & 1 \\
 1 & 1 & -1 & -1 & 1 & 1 & -1 & -1 \\
 -1 & -1 & -1 & -1 & 1 & 1 & 1 & 1 \\
 1 & 1 & 1 & 1 & 1 & 1 & 1 & 1 \\
    \end{pmatrix}.
\end{equation}
For example, when $m_x=0.7$, $t_x=0.1$, $\gamma_x=0.8$, $t_y=t_z=0$, $\gamma_y=\gamma_z=1$, this system exhibits a HOTP with the configuration vector $\v{\chi}_8=(-1,1,0,0,0,0,1,-1)$. As shown in Fig.~\ref{fig:Sup-Fig1_penta_3d}, the Bott index vector $\v{\nu}_8$ equals $(2,0,0,0, -2,0,0)$, demonstrating that HOTPs can be fully characterized by
\begin{equation}
     \boldsymbol{\chi}_{8}=\mathcal{M}^{-1}\cdot \left(\nu^{(1)}_8,\dots,\nu^{(7)}_8,0\right)^{\mathrm{T}}.
\end{equation}

\bigskip

\def\bibsection{\ } 
\noindent \textbf{REFERENCES}

\bibliographystyle{naturemag}
\bibliography{ref}
\end{document}